\documentclass{article}
\usepackage[left=1in, right=1in, top=1in, bottom=1in]{geometry}
\usepackage{enumitem}
\usepackage{authblk}
\usepackage{graphicx}	
\usepackage{amsmath, amsthm, amssymb}
\usepackage{xspace}
\usepackage{comment}
\usepackage{hyperref}
\usepackage[dvipsnames]{xcolor}
\usepackage{algorithm}
\usepackage{algpseudocode}
\usepackage{todonotes}
\usepackage{mathtools}
\usepackage[capitalize]{cleveref}
\usepackage{ifthen}
\usepackage[mathlines]{lineno}
\usepackage[normalem]{ulem} 

\newcommand{\algprobm}[1]{\textsc{#1}\xspace}

\newcommand{\betacc}[1]{\ifthenelse{\equal{#1}{1}}{\exists^{\log n}}{\exists^{\log^{#1}n}}} 
\newcommand{\alphacc}[1]{\ifthenelse{\equal{#1}{1}}{\forall^{\log n}}{\forall^{\log^{#1}n}}} 

\newcommand{\Soc}{\text{Soc}}

\def\dis#1{\mathrm{Dis}(#1)}
\def\ldiv{\backslash}
\def\rdiv{/}

\renewcommand{\setminus}{\mysetminus}
\newcommand{\mysetminusD}{\raisebox{.8pt}{\hbox{\tikz{\draw[line width=0.6pt,line cap=round] (3.5pt,0pt) -- (0,5.2pt);}}}}
\newcommand{\mysetminusT}{\mysetminusD}
\newcommand{\mysetminusS}{\raisebox{.5pt}{\hbox{\tikz{\draw[line width=0.45pt,line cap=round] (2.2pt,0) -- (0,3.8pt);}}}}
\newcommand{\mysetminusSS}{\raisebox{.35pt}{\hbox{\tikz{\draw[line width=0.4pt,line cap=round] (1.5pt,0) -- (0,2.8pt);}}}}

\newcommand{\mysetminus}{\mathbin{\mathchoice{\mysetminusD}{\mysetminusT}{\mysetminusS}{\mysetminusSS}}}

\theoremstyle{plain}
\newtheorem{theorem}{Theorem}[section]
\newtheorem{proposition}[theorem]{Proposition}
\newtheorem{corollary}[theorem]{Corollary}
\newtheorem{lemma}[theorem]{Lemma}

\theoremstyle{definition}
\newtheorem{definition}[theorem]{Definition}
\newtheorem{remark}[theorem]{Remark}
\newtheorem{question}[theorem]{Question}

\newcommand{\wt}{\text{wt}}

\DeclareMathOperator{\Aut}{Aut}

\DeclareMathOperator{\poly}{poly}

\DeclareMathOperator{\rk}{rk}

\newcommand*{\ComplexityClass}[1]{\ensuremath{\mathsf{#1}}\xspace}

\newcommand{\DTISP}{\ComplexityClass{DTISP}}

\newcommand{\DTISPpll}{\ensuremath{\DTISP(\mathrm{polylog}(n),\mathrm{log}(n))}\xspace}

\newcommand*{\LogSpace}{\ComplexityClass{L}}

\newcommand{\AC}{\ComplexityClass{AC}}

\newcommand{\ACz}{\ComplexityClass{AC^0}}

\newcommand{\FOLL}{\ComplexityClass{FOLL}}

\title{On the Parallel Complexity of Identifying Groups and Quasigroups via Decompositions\footnote{A preliminary version of this work appeared in the proceedings of RAMiCS 2026 \cite{JLVWQuasigroups}. In this version, we provide the full proofs that were omitted from the conference version, due to space constraints.}}

\author[1]{Dan Johnson}
\author[1]{Michael Levet}
\author[2]{Petr Vojtěchovsk\'y}
\author[1]{Brett Widholm}

\affil[1]{Department of Computer Science, College of Charleston}
\affil[2]{Department of Mathematics, University of Denver}

\begin{document}
\maketitle
\begin{abstract}
In this paper, we investigate the computational complexity of isomorphism testing for finite groups and quasigroups, given by their multiplication tables. We crucially take advantage of their various decompositions to show the following:
\begin{itemize}
\item We first consider the class of groups that admit direct product decompositions, where each indecompsable factor is $O(1)$-generated, and either perfect or centerless. We show any group in this class is identified by the $O(1)$-dimensional count-free Weisfeiler--Leman (WL) algorithm with $O(\log \log n)$ rounds, and the $O(1)$-dimensional counting WL algorithm with $O(1)$ rounds. Consequently, the isomorphism problem for this class is in $\textsf{L}$. This improves upon the  previous upper bound of $\textsf{TC}^{1}$, which was obtained using $O(\log n)$ rounds of the $O(1)$-dimensional counting WL (Grochow and Levet; FCT 2023, \textit{J. Comput. Syst. Sci.} 2026).

\item We next consider more generally, the class of groups where each indecomposable factor is $O(1)$-generated. We exhibit an $\textsf{AC}^{3}$ canonical labeling procedure for this class. Here, we accomplish this by showing that in the multiplication table model, the direct product decomposition can be computed in $\textsf{AC}^{3}$, parallelizing the work of Kayal and Nezhmetdinov (ICALP 2009).

\item Isomorphism testing between a central quasigroup $G$ and an arbitrary quasigroup $H$ is in $\textsf{NC}$. Here, we take advantage of the fact that central quasigroups admit an affine decomposition in terms of an underlying Abelian group. Only the trivial bound of $n^{\log(n)+O(1)}$-time was previously known for isomorphism testing of central quasigroups.
\end{itemize}
\end{abstract}

\thispagestyle{empty}

\newpage

\setcounter{page}{1}

\section{Introduction}

The \algprobm{Quasigroup Isomorphism} (\algprobm{QGpI}) problem takes as input two finite quasigroups $G, H$ and asks if there exists an isomorphism $\varphi : G \to H$. When the quasigroups are given by their multiplication (Cayley) tables, it is known that $\algprobm{QGpI}$ belongs to $\textsf{NP} \cap \textsf{coAM}$ (see e.g., \cite{GoldreichMicaliWidgerson}). The generator-enumeration strategy has time complexity $n^{\log(n) + O(1)}$, where $n$ is the order of the group. In the setting of groups, this strategy was independently discovered by Felsch and Neub\"user \cite{FN} and Tarjan (see \cite{MillerTarjan}). Miller subsequently extended this strategy to the setting of quasigroups \cite{MillerTarjan}. The parallel complexity of the generator-enumeration strategy has been gradually improved \cite{LiptonSnyderZalcstein, Wolf, ChattopadhyayToranWagner, WagnerThesis, TangThesis}. Collins, Grochow, Levet, and Weiß recently improved the bound for the generator enumeration strategy to
\[
\exists^{\log^{2} n} \forall^{\log n} \exists^{\log n}\textsf{DTISP}(\text{polylog}(n), \log(n)),
\]
which can be simulated by depth-$4$ $\textsf{quasiAC}^{0}$ circuits of size $n^{O(\log n)}$ \cite{CGLWISSAC}. 

In the special case of \algprobm{Group Isomorphism} (\algprobm{GpI}), the worst-case runtime has been improved to $n^{(1/4)\log(n)+O(1)}$, due to Rosenbaum \cite{Rosenbaum2013BidirectionalCD} and Luks \cite{LuksCompositionSeriesIso} (see \cite[Sec. 2.2]{GR16}). Even the impressive body of work on practical algorithms for \algprobm{GpI} in more succinct input models, led by Eick, Holt, Leedham-Green and O'Brien (e.\,g., \cite{BEO02, ELGO02, BE99, CH03}) still results in an $n^{\Theta(\log n)}$-time algorithm in the general case (see \cite[Page 2]{WilsonSubgroupProfiles}).

In practice, such as working with computer algebra systems, the Cayley model is highly unrealistic. Instead, the groups are given by generators as permutations or matrices, or as black-boxes. In the setting of permutation groups, \algprobm{GpI} belongs to $\textsf{NP}$ \cite{LuksReduction}. When the groups are given by generating sets of matrices, or as black-boxes, $\algprobm{GpI}$ belongs to $\textsf{Promise}\Sigma_{2}^{p}$ \cite{BabaiSzemeredi}; it remains open as to whether $\algprobm{GpI}$ belongs to $\textsf{NP}$ or $\textsf{coNP}$ in such succinct models.  In the past several years, there have been significant advances on algorithms with worst-case guarantees on the serial runtime for special cases of this problem; we refer to \cite{GQCoho, DietrichWilson, GrochowLevetWL} for a survey. The work on generator-enumeration has resulted in an $\textsf{SAC}^{1}$-bound for isomorphism testing of $O(1)$-generated quasigroups \cite{WagnerThesis}. Otherwise, to the best of our knowledge, polynomial-time isomorphism tests for quasigroups that are not groups have not been previously investigated.

Key motivation for \algprobm{QGpI} comes from its relationship to the \algprobm{Graph Isomorphism} problem (\algprobm{GI}). When the quasigroups are given verbosely by their multiplication tables, there exists an $\textsf{AC}^{0}$-computable many-one reduction from \algprobm{QGpI} to $\algprobm{GI}$ \cite{ZKT}. On the other hand, there is no reduction from \algprobm{GI} to \algprobm{QGpI} computable by $\textsf{AC}$ circuits of depth $o(\log n / \log \log n)$ and size $n^{\text{polylog}(n)}$ \cite{ChattopadhyayToranWagner}. In light of Babai's breakthrough result that $\algprobm{GI}$ is quasipolynomial-time solvable \cite{BabaiGraphIso}, $\algprobm{QGpI}$ in the Cayley model is a key barrier to improving the complexity of $\algprobm{GI}$. There is considerable evidence suggesting that $\algprobm{GI}$ is not $\textsf{NP}$-complete \cite{Schoning, BuhrmanHomer, ETH, BabaiGraphIso, GILowPP, ArvindKurur}. As $\algprobm{QGpI}$ reduces to $\algprobm{GI}$, this evidence also suggests that $\algprobm{QGpI}$ is not $\textsf{NP}$-complete.

Despite the fact that \algprobm{QGpI} is strictly easier than \algprobm{GI} under $\ACz$-reductions, there are several key approaches in the \algprobm{GI} literature, such as parallelization, Weisfeiler--Leman, and canonization, that have received comparatively little attention even in the special case of \algprobm{GpI}-- see the discussion of Further Related Work (Section~\ref{sec:RelatedWork}). In this paper, we will make advances in each of these directions.

There is a natural divide-and-conquer strategy for \algprobm{GpI}, which proceeds as follows. Given groups $G_1, G_2$, we first seek to decompose the groups in terms of normal subgroups $N_i \trianglelefteq G_i$ ($i \in [2]$) and the corresponding quotients $G_i/N_i$. It is then necessary (but not sufficient) to decide if $N_1 \cong N_2$ and $G_1/N_1 \cong G_2/N_2$. This strategy has led to algorithmic advances for groups admitting \textit{nice} decompositions, including for instance, direct product decompositions \cite{KayalNezhmetdinov, WilsonDirectProductsArxiv, BrachterSchweitzerWLLibrary, GrochowLevetWL}, coprime \cite{LeGall, QST11, BQ, GrochowLevetWL} and tame \cite{GQ15} extensions, and central extensions \cite{LewisWilson, BMWGenus2, IvanyosQ19, BGLQW,GQCoho}. In case of quasigroups the situation is made more complicated by the lack of normal substructures; one has to factor by congruences instead. (Recall that a \emph{congruence} of an algebra $A$ is an equivalence relation on $A$ that is compatible with all operations of $A$.) We therefore focus on the isomorphism problem for quasigroups that admit representations in terms of more structured algebras, such as groups and their automorphisms.

\noindent \\ \textbf{Main Results.} In this paper, we investigate the parallel complexity of isomorphism testing for finite groups and quasigroups, by taking advantage of key decompositions.

We first investigate the power of the Weisfeiler--Leman (WL) algorithm to identify groups based on their direct product decompositions. Weisfeiler--Leman is a key combinatorial algorithm used in the setting of \algprobm{Graph Isomorphism}. Brachter and Schweitzer recently provided three adaptations to the setting of groups, specified by their multiplication (Cayley) tables \cite{WLGroups}. In this work, we will use only the first two adaptations, which Brachter and Schweitzer refer to as WL Versions I and II. We will provide a brief overview of WL here, and we refer to Section~\ref{sec:WL} for a precise formulation. For $k \geq 2$, the $k$-dimensional Weisfeiler--Leman algorithm iteratively colors $k$-tuples of group elements in an isomorphism-invariant manner. At each refinement step, each $k$-tuple receives a new color based on (i) its current color, and (ii) the multiset of colors from the ``nearby" $k$-tuples. Additionally, we will be interested in the weaker \emph{count-free} variant of Weisfeiler--Leman, in which the refinement steps only take into account the set of colors from the ``nearby" $k$-tuples, rather than the full multiset. Grohe and Verbitsky \cite{GroheVerbitsky} observed that WL can be effectively parallelized (see Section~\ref{sec:WLparallel} for further discussion). Consequently, controlling both the dimension and number of iterations can yield improvements in the parallel complexity of isomorphism testing.

We now turn to discussing direct product decompositions. Kayal and Nezhmetdinov \cite{KayalNezhmetdinov} and Wilson \cite{WilsonDirectProductsArxiv} previously exhibited polynomial-time algorithms to compute direct product decompositions of finite groups into indecomposable direct factors (we refer to such decompositions as \emph{fully refined}-- see Section~\ref{sec:DirectProductPreliminaries} for more background on direct products). For group isomorphism testing, the main benefit of these results is that, for the purposes of getting polynomial-time isomorphism tests, they reduce the problem to the case of directly indecomposable groups.

Brachter and Schweitzer \cite{BrachterSchweitzerWLLibrary} established a WL-analogue, by relating the Weisfeiler--Leman dimension for an arbitrary group $G$, to that of $G$'s indecomposable direct factors. Grochow and Levet \cite{GrochowLevetWL} parallelized this result, essentially showing that WL can detect the indecomposable direct factors using only $O(\log n)$ rounds. While WL does not explicitly return a direct product decomposition itself, for the purposes of isomorphism testing, WL achieves essentially the same benefit of reducing to the directly indecomposable case-- in $\textsf{TC}^{1}$ (at the cost of increasing the dimension by $1$ and adding $O(\log n)$ rounds) \cite{GrochowLevetWL}. 

On the other hand, Grohcow and Levet (\emph{ibid}.) showed that the analogous result fails to hold for the \emph{count-free} variant of WL. While cyclic groups are identified by the count-free $2$-WL, there exist infinite families $(G_m, H_m)_{m \in \mathbb{N}}$ of non-isomorphic Abelian groups that require $\Omega(\log |G_m|)$-dimensional (and hence, $\Theta(\log |G_m|)$-dimensional) count-free WL to distinguish between $G_m$ and $H_m$. 

Motivated by this gap, we investigate the ability of count-free WL to identify groups that admit \textit{nice} direct product decompositions. The Remak--Krull--Schmidt theorem provides that for any two fully-refined direct product decompositions $\mathcal{S} = \{ S_1, \ldots, S_k\}$ and $\mathcal{T} = \{T_1, \ldots, T_k\}$ of a group $G$, there exists a permutation $\pi \in \text{Sym}(k)$ such that $S_{i} \cong T_{\pi(i)}$ for all $i \in [k]$. In general, however, $S_{i}$ and $T_{\pi(i)}$ need not be the same subgroup (setwise). On the other hand, it is well-known that when $S_i$ is centerless, it appears (setwise) in every fully-refined direct product decomposition of $G$. The same holds true when $S_i$ is perfect (see Corollary~\ref{cor:PerfectCenterless}). For this reason, we investigate the distinguishing power of the count-free Weisfeiler--Leman algorithm for the following class of groups.

\begin{definition} \label{def:PerfectCenterless}
Let $\mathcal{C}$ be the class of groups $G$, where $G$ admits a direct product decomposition into indecomposable direct factors $G = \prod_{i=1}^{k} S_{i}$, such that each $S_{i}$ is (i) $O(1)$-generated, and (ii) either perfect or centerless. 
\end{definition}

Note that $\mathcal{C}$ includes, for instance, direct products of almost simple groups.

As controlling the WL-dimension and the number of rounds can yield improvements in the parallel complexity for isomorphism testing \cite{GroheVerbitsky} (see Section~\ref{sec:WLparallel} for more discussion), we will use the following notation. Let $k \geq 2, r \geq 0$, and let $(k,r)$-WL denote running $k$-dimensional Weisfeiler--Leman for $r$ rounds.

\noindent Our first main result is the following.

\begin{theorem} \label{thm:MainWL}
Let $G \in \mathcal{C}$ be a group of order $n$.
\begin{enumerate}[label=(\alph*)]
\item The count-free $(O(1), O(\log \log n))$-WL Version II algorithm identifies $G$. 

\item The counting $(O(1), O(1))$-WL Version II algorithm identifies $G$.
\end{enumerate}

\noindent In particular, the constants hidden by the Big-O terms depend only on $\mathcal{C}$.
\end{theorem}

We refer to Section~\ref{sec:WL} for a precise formulation of WL Versions I and II.

Grochow and Levet previously established that the counting $(O(1), O(\log n))$-WL Version II identifies the groups in $\mathcal{C}$. In light of the parallel WL implementation due to Grohe and Verbitsky \cite{GroheVerbitsky} (see also Section~\ref{sec:WLparallel}), Grochow and Levet obtained that isomorphism testing between a group $G \in \mathcal{C}$ and an arbitrary group $H$ is in $\textsf{TC}^{1}$, which improved upon the previous bound of $\textsf{P}$. Theorem~\ref{thm:MainWL}(a) already improves this bound from $\textsf{TC}^{1}$ to $\textsf{AC}^{1}$. 

Theorem~\ref{thm:MainWL}(b) is obtained by a more careful analysis of the counting WL algorithm that takes advantage of the group-theoretic structure afforded by $\mathcal{C}$-- namely, that the indecomposable direct factors are unique as sets. In addition to making precise the power of counting in this situation (reducing the number of rounds from $O(\log \log n)$ to $O(1)$), we obtain further improvements in the parallel complexity for isomorphism testing:

\begin{corollary} \label{cor:ParallelWL}
There exists a logspace algorithm that, given $G \in \mathcal{C}$ and an arbitrary group $H$, correctly decides if $G \cong H$.
\end{corollary}

\begin{remark}
It remains an intriguing open question whether Theorem~\ref{thm:MainWL} can be extended to the class of groups that admit a unique direct product decomposition, where each indecomposable direct factor is $O(1)$-generated. Let $G$ and $H$ be groups. Let $\mathcal{S} = \{ S_1, \ldots, S_k\}$ be a fully-refined direct product decomposition for $G$, and let $\mathcal{T} = \{ T_1, \ldots, T_{\ell}\}$ be a fully-refined direct product decomposition for $H$. 

Note that $\mathcal{S}$ is unique if and only if, for all distinct $i,j \in [k]$, $|S_{i}/[S_{i}, S_{i}]|$ and $|Z(S_{j})|$ are coprime (Theorem~\ref{thm:UniqueDecomposition}). Now Weisfeiler--Leman can detect whether $G$ and $H$ contain a normal subgroup isomorphic to some $S_{i}$ ($i \in [k]$), provided $S_{i}$ is $O(1)$-generated. Precisely, in the pebble game characterization, Spoiler can pebble generators $g_1, \ldots, g_d$ for $S_{i}$. Let $h_1, \ldots, h_d$ be Duplicator's response. If $T := \langle h_1, \ldots, h_d \rangle$ is either not isomorphic to $S$ or is not normal in $H$, then Spoiler can win (see e.g., Lemma~\ref{lem:RankOne} for full details).

Now $T$ is perfect (resp., centerless) if and only if $S_{i}$ is also perfect (resp. centerless). Thus, $T$ appears in every fully-refined direct product decomposition of $H$ (Corollary~\ref{cor:PerfectCenterless}). Without loss of generality, let $T := T_{1} \in \mathcal{T}$. It is unclear how to certify, in the pebble game, whether for some $j > 1$, $|T_{1}/[T_{1}, T_{1}]|$ and $|Z(T_{j})|$ are coprime. This is the obstacle towards extending Theorem~\ref{thm:MainWL} to the class of groups that admit a unique direct product decomposition, where each indecomposable direct factor is $O(1)$-generated.
\end{remark}

We also compare Theorem~\ref{thm:MainWL}(a) to the work of Collins and Levet \cite{CollinsLevetWL}, who showed that count-free WL can be fruitfully leveraged as a subroutine for isomorphism testing of direct products of non-Abelian simple groups. Note that such groups all belong to $\mathcal{C}$. Collins and Levet showed that the count-free $(O(1), O(\log \log n))$-WL Version I will distinguish between two non-isomorphic direct factors. They then used $O(\log n)$ universally quantified co-nondeterministic bits and a single $\textsf{Majority}$ gate to verify the multiplicities of the isomorphism type for each direct factor appeared the same number of times in both groups. Following the notation of \cite{CGLWISSAC}, we refer to this corresponding complexity class as $\forall^{\log n}\textsf{MAC}^{0}(\textsf{FOLL})$. Theorem~\ref{thm:MainWL}(a) shows that count-free WL \emph{on its own} can identify this class of groups. On the other hand, the bound obtained by Collins and Levet of $\forall^{\log n}\textsf{MAC}^{0}(\textsf{FOLL})$ is stronger than the bound of $\textsf{AC}^{1}$ afforded by Theorem~\ref{thm:MainWL}(a). This brings us to our second main result.

\begin{theorem} \label{thm:MainSimple}
The count-free $(O(1), O(\log \log n))$-WL Version I algorithm identifies every group $G$ that decomposes as a direct product of non-Abelian simple groups. 
\end{theorem}

While the initial coloring of count-free WL Version II is $\textsf{L}$-computable \cite{TangThesis}, the initial coloring of count-free WL Version I is $\textsf{AC}^{0}$-computable (see Section~\ref{sec:WLparallel} for more details).  Consequently, we obtain the following corollary.

\begin{corollary}
There exists a uniform \FOLL algorithm that decides isomorphism between a group $G$ that decomposes as a direct product of non-Abelian simple groups, and an arbitrary group $H$.   
\end{corollary}

As \FOLL cannot compute \algprobm{Parity} or \algprobm{Majority} \cite{FSS}, we have that $\FOLL \subsetneq \forall^{\log n}\textsf{MAC}^{0}(\textsf{FOLL}) \subseteq \textsf{AC}^{1}$.  

\begin{remark}
In order to obtain $O(\log \log n)$ rounds in Theorem~\ref{thm:MainSimple}, we take advantage of a deep consequence of the Classification of Finite Simple Groups, that there exists absolute constants $C_1, C_2$, such that for every finite simple group $S$, there exists a generating set $X$ of size $C_1$ such that 
\[
\text{diam}(\text{Cay}(S, X)) \leq C_{2} \log |S|.
\]
Babai, Kantor, and Lubotzky \cite{BabaiKantorLubotsky} showed that we can take $C_1 = 7$. Subsequently, Kassabov and Riley \cite{KassabovRiley} showed that we can take $C_1 = 4$. We are not aware of similar results for $O(1)$-generated perfect or centerless groups. In particular, extending Theorem~\ref{thm:MainWL} to use WL Version I will likely require new techniques. We leave this as an open question.
\end{remark}

Our next result concerns the canonization problem. For a class $\mathcal{K}$ of objects, a \textit{canonical form} is a function $F : \mathcal{K} \to \mathcal{K}$, such that for all $X, Y \in \mathcal{K}$, we have (i) $X \cong F(X)$, and (ii) $X \cong Y \iff F(X) = F(Y).$ Isomorphism testing reduces to canonization. In the general settings of graphs and groups, the converse remains open. Nonetheless, efficient canonization procedures have often followed the corresponding efficient isomorphism tests, usually with non-trivial work-- see e.g., \cite{ImmermanLander1990, BabaiLuksCanonicalLabeling, KoblerVerbitsky, BabaiQuasipolynomialCanonization, ElberfeldSchweitzer}.

We will, in particular, design an algorithm that computes canonical labelings. Precisely, if $X, Y \in \mathcal{K}$ have order $n$, our algorithm will return labelings $\lambda_{X} : X \to [n]$ and $\lambda_{Y} : Y \to [n]$ such that $X \cong Y$ if and only if $\lambda_{Y}^{-1} \circ \lambda_{X}$ is an isomorphism between $X$ and $Y$. We establish the following:

\begin{theorem} \label{thm:MainCanonization}
Let $d > 0$ be a constant. There exists a uniform $\textsf{AC}^{3}$ algorithm such that the following holds. Let $G$ be a group of order $n$ that admits a fully-refined direct product decomposition such that each indecomposable direct factor is $d$-generated. Our algorithm computes a canonical labeling of $G$.
\end{theorem}

Note that the groups considered in Theorem~\ref{thm:MainCanonization} include those in $\mathcal{C}$. 

We establish Theorem~\ref{thm:MainCanonization} in two steps. First, we will compute a fully-refined direct product decomposition of $G$. Second, we will use $(O(1), O(1))$-WL Version II on $G$, which will in turn allow us to pick out a canonical generating set for each indecomposable direct factor, and hence compute a canonical labeling for each direct factor. We also note that the color classes from WL are labeled using numbers (see e.g., \cite{GroheVerbitsky}). This will in turn allow us to order the indecomposable direct factors.

In order to explicitly compute a fully-refined direct product decomposition, we will establish the following:

\begin{theorem}[cf. {\cite{KayalNezhmetdinov}}] \label{thm:MainDecompose}
Let $G$ be a finite group given by its multiplication table. We can compute a fully-refined direct product decomposition of $G$ in $\textsf{AC}^{3}$.
\end{theorem}

Previously, Kayal and Nezhmetdinov \cite{KayalNezhmetdinov} established polynomial-time bounds for computing a fully-refined direct product decomposition of $G$. Wilson \cite{WilsonDirectProductsArxiv} established the analogous result in the succinct \emph{quotients of permutation groups} model. As a consequence, Wilson obtained polynomial-time bounds for computing a fully-refined direct product decomposition in the multiplication table model, by considering the regular representation of the input group. We establish Theorem~\ref{thm:MainDecompose} by parallelizing the work of Kayal and Nezhmetdinov \cite{KayalNezhmetdinov}. 

\begin{remark}
It remains an intriguing open question as to whether the problem of computing a fully-refined direct product decomposition admits an $\textsf{NC}$ solution in the setting of quotients of permutation groups. There are a number of  obstacles outlined in \cite[Section~8]{WilsonDirectProductsArxiv}. Additionally, it is open whether certain key problems such as computing the centralizer $C_{G}(N)$, for $N \trianglelefteq G$, admit $\textsf{NC}$ solutions in the setting of quotients of permutation groups. Many of the subroutines for working with quotients of permutation groups rely on the fact that we can, in polynomial-time, compute $X \cap P$, where $X$ is an arbitrary permutation group and $P$ is a permutation group of prime power order \cite{LuksBoundedValence}. The existence of an $\textsf{NC}$ solution for this problem is a longstanding open problem.
\end{remark}

We finally consider central quasigroups.

\begin{definition} \label{def:CentralQuasigroup}
A quasigroup $(Q,*)$ is \emph{central} if there is an Abelian group $(Q,+)$, automorphisms $\phi$, $\psi$ of $(Q,+)$ and $c\in Q$ such that
\begin{equation}\label{Eq:Central}
    x*y = \phi(x)+\psi(y)+c
\end{equation}
for all $x,y\in Q$. We denote the corresponding quasigroup by $\mathcal Q(Q,+,\phi,\psi,c)$.  
\end{definition}

Central quasigroups play an important role since they are precisely the Abelian objects (in the sense of universal algebra) in the variety of quasigroups \cite{Szendrei}. See \cite{KepkaNemec1, KepkaNemec2, Drapal, SmithBook, SmithChapter} for an introduction to central quasigroups and \cite{Hou, Kirnasovsky, StanovskyVojtechovsky} for enumeration of certain classes of central quasigroups up to isomorphism.

We establish the following.
 
\begin{theorem} \label{thm:MainCentralQuasigroups}
Let $G_1$ be a central quasigroup, and let $G_2$ be an arbitrary quasigroup. Suppose that $G_1, G_2$ are given by their multiplication tables. We can decide isomorphism between $G_1$ and $G_2$ in $\textsf{NC}$.        
\end{theorem}

In contrast, the complexity of isomorphism testing for Abelian groups has been well-investigated \cite{LiptonSnyderZalcstein, Vikas, Savage, Kavitha, ILIOPOULOS198581, VillardSNF, ChattopadhyayToranWagner, GrochowLevetWL, CGLWISSAC}, resulting in a linear time serial algorithm \cite{Kavitha} and parallel bounds of $\forall^{\log \log n}\textsf{MAC}^{0}(\DTISPpll)$, the latter of which is a proper subclass of $\LogSpace$ \cite{CGLWISSAC}.

A quasigroup $(Q,*)$ is \emph{medial} (also \emph{entropic}) if it satisfies the medial law, for all $x, y, u, v \in Q$:
\begin{displaymath}
    (x*y)*(u*v) = (x*u)*(y*v).
\end{displaymath}

It was shown by Toyoda--Murdoch--Bruck \cite{Toyoda, Murdoch, Bruck} that medial quasigroups form a subclass of central quasigroups. In more detail, up to isomorphism, medial quasigroups are precisely the central quasigroups $\mathcal Q(Q,+,\phi,\psi,c)$ such that $\phi\psi=\psi\phi$. A geometric interpretation of this representation can be found in \cite{Volenec}. Small entropic quasigroups were enumerated in \cite{StanovskyVojtechovsky}. The medial law has also been studied more generally in \cite{Bonatto, JezekKepka}.

\subsection{Further Related Work} \label{sec:RelatedWork}

\paragraph{Parallelization.} There is a long history of work on developing efficient parallel ($\textsf{NC}$) algorithms for special cases of \algprobm{Graph Isomorphism}. We refer to \cite{LevetRombachSieger} for a survey. In comparison, the work on $\textsf{NC}$ algorithms for \algprobm{(Quasi)Group Isomorphism} is comparatively nascent. Much of the work revolves around parallelizing the generator-enumeration strategy (\emph{ibid}.), culminating in bounds of $\textsf{L}$ for $O(1)$-generated groups \cite{TangThesis} and $\textsf{SAC}^{1}$ for $O(1)$-generated quasigroups \cite{WagnerThesis}. Prior to this paper (Theorem~\ref{thm:MainCentralQuasigroups}), no other special family of quasigroups (beyond certain families of groups) was known to even admit an polynomial-time isomorphism test. Other families of groups known to admit $\textsf{NC}$ isomorphism tests include Abelian groups \cite{VillardSNF, ChattopadhyayToranWagner, GrochowLevetWL, CGLWISSAC}, graphical groups arising from the CFI graphs \cite{WLGroups, CollinsLevetWL, CollinsUndergradThesis}, coprime extensions $H \ltimes N$ where $H$ is $O(1)$-generated and $N$ is Abelian \cite{GrochowLevetWL} (parallelizing a result from \cite{QST11}), groups of almost all orders \cite{CGLWISSAC} (parallelizing \cite{DietrichWilson}), and Fitting-free groups \cite{GroochowJohnsonLevet} (parallelizing \cite{BCQ}).

\paragraph{Weisfeiler--Leman.} The $k$-dimensional Weisfeiler--Leman algorithm ($k$-WL) serves as a key combinatorial tool in \textsc{GI}. It works by iteratively coloring $k$-tuples of vertices in an isomorphism-invariant manner. On its own, Weisfeiler--Leman serves as an efficient polynomial-time isomorphism test for several families of graphs-- see Kiefer's dissertation for a survey \cite{KieferThesis}. In the case of graphs of bounded treewidth~\cite{GroheVerbitsky, LevetRombachSieger}, planar graphs~\cite{GroheVerbitsky, VerbitskyPlanar,  GroheKieferPlanar}, and graphs of bounded rank-width \cite{LevetRombachSieger}, Weisfeiler--Leman serves even as a $\textsf{TC}^{1}$ isomorphism test. Despite the success of WL as an isomorphism test, it is insufficient to place \textsc{GI} into $\textsf{P}$~\cite{CFI, NeuenSchweitzerIR}. Nonetheless, WL remains an active area of research. For instance, Babai's quasipolynomial-time algorithm~\cite{BabaiGraphIso} combines $O(\log n)$-WL with group-theoretic techniques.

The use of combinatorial techinques for $\algprobm{GpI}$ is relatively new; and to the best of our knowledge, was initiated in \cite{QiaoLiWL, BGLQW}. Brachter and Schweitzer \cite{WLGroups} subsequently initiated the study  of Weisfeiler--Leman applied directly to the Cayley table. In particular, they established that graphical groups \cite{Mekler, HeQiao} arising from the CFI graphs \cite{CFI} have WL-dimension at most $3$. There has been considerable work building on the framework introduced by Brachter and Schweitzer-- see \cite{BrachterSchweitzerWLLibrary, BrachterThesis, GLDescriptiveComplexity, GrochowLevetWL, CollinsLevetWL, CollinsUndergradThesis, ChenRenPonomarenko, VagnozziThesis, GroochowJohnsonLevet}. It remains open whether the class of finite groups has WL-dimension $o(\log n)$.

\paragraph{Canonization.} It is open whether canonization reduces to isomorphism testing. Nonetheless, in the setting of graphs, efficient canonization procedures have often followed efficient isomorphism tests, usually with non-trivial work-- see e.g., \cite{ImmermanLander1990, grohe2019canonisation,  KoblerVerbitsky, WagnerBoundedTreewidth, ElberfeldSchweitzer, BabaiQuasipolynomialCanonization, LevetRombachSieger}. Less is known about canonization in the setting of groups and quasigroups. Abelian groups admit an $\textsf{NC}^{2}$ canonization procedure using Smith normal form \cite{VillardSNF}. Additionally, canonization for $O(1)$-generated groups is in $\textsf{L}$ (see e.g., \cite{LevetThesis}). In the special case of finite simple groups, canonization is $\FOLL$ \cite{CollinsLevetWL}. Gill, Mammoliti, and Wanless showed that quasigroups admit a canonization procedure that has polynomial-time average-case complexity \cite{GillMammolitiWanless, GillThesis}.

\section{Preliminaries}

\subsection{Groups and Quasigroups}

All groups and quasigroups will be assumed to be finite.

\paragraph{Groups.} Let $G$ be a group. Let $Z(G)$ denote the center of $G$. Given $g, h \in G$, the \textit{commutator} $[g, h] := ghg^{-1}h^{-1}$. For sets $X, Y \subseteq G$, the \textit{commutator subgroup} $[X,Y] := \langle \{ [g, h] : g \in X, h \in Y \} \rangle$. We say that $G$ is \textit{perfect} if $G = [G, G]$, and that $G$ is \textit{centerless} if $Z(G) = 1$. Let $d(G)$ be the minimum size taken over all generating sets of $G$. Fix $d \in \mathbb{N}$; we say that a class $\mathcal{K}$ of groups is \emph{$d$-generated} if for each group $G \in \mathcal{K}$, $d(G) \leq d$. A \textit{basis} of an Abelian group $A$ is a set of generators $\{a_1, \ldots, a_k\}$ such that $A = \langle a_1 \rangle \times \cdots \times \langle a_k \rangle.$ We say that a subgroup $N \leq G$ is \emph{characteristic} if for all $\alpha \in \Aut(G)$, $\alpha(N) = N$.

For a group $G$ and a subset $S \subseteq G$, with $1 \not \in S$ and $S = S^{-1}$, the \emph{Cayley graph} $\text{Cay}(G,S)$ has vertex set $G$. There is an edge $\{g,h\}$ in $\text{Cay}(G,S)$ if and only if there exists some $s \in S$ such that $gs = h$. 

\paragraph{Quasigroups.} A \textit{quasigroup} consists of a set $Q$ and a binary operation $*: Q \times Q \to Q$ satisfying the following. For every $a, b \in Q$, there exist unique $x, y\in Q$ such that $a * x = b$ and $y * a = b$.  We write $x = a\ldiv b$ and $y = b\rdiv a$, i.e., $a*(a\ldiv b) = b$ and $ (b\rdiv a)*a = b$. When the multiplication operation is understood, we simply write $ax$ for $a*x$.

If $(Q_1,*)$, $(Q_2,\circ)$ are quasigroups, then $(\alpha,\beta,\gamma)$ is an \emph{isotopism} from $(Q_1,*)$ to $(Q_2,\circ)$ if $\alpha,\beta,\gamma:Q_1\to Q_2$ are bijections and $\alpha(x)\circ\beta(y) = \gamma(x*y)$ for all $x,y\in Q_1$. Being isotopic is an equivalence relation on quasigroups. The following result is well known: 

\begin{proposition}[{\cite{Albert}}]
If two groups are isotopic, they are isomorphic. In particular, if a quasigroup $(Q,*)$ is isotopic to groups $G_1$ and $G_2$, then $G_1$ is isomorphic to $G_2$.
\end{proposition}

For $x \in (Q, *)$, let $L_{x} : Q \to Q$ be the \textit{left translation} by $x$, where $L_{x}(y) = x * y$. For a quasigroup $(Q, *)$, let 
\begin{displaymath}
    \dis{Q,*} = \langle L_{x}L_{y}^{-1} : x,y\in Q\rangle.
\end{displaymath}
The group $\dis{Q,*}$ acts naturally on $Q$, being a subgroup of the symmetric group on $Q$. For any $e\in Q$, note that $\dis{Q,*} = \langle L_xL_e^{-1}:x\in X\rangle$, since $L_xL_y^{-1} = (L_xL_e^{-1})(L_yL_e^{-1})^{-1}$.

\subsection{Complexity Classes} \label{sec:Complexity}
We assume that the reader is familiar with standard complexity classes such as $\textsf{P}, \textsf{NP}, \textsf{L}$, and $\textsf{NL}$. For a standard reference on circuit complexity, see \cite{VollmerText}. We consider Boolean circuits using the gates \textsf{AND}, \textsf{OR}, \textsf{NOT}, and \textsf{Majority}, where $\textsf{Majority}(x_{1}, \ldots, x_{n}) = 1$ if and only if $\geq n/2$ of the inputs are $1$. Otherwise, $\textsf{Majority}(x_{1}, \ldots, x_{n}) = 0$. In this paper, we will consider $\textsf{DLOGTIME}$-uniform circuit families $(C_{n})_{n \in \mathbb{N}}$. For this,
one encodes the gates of each circuit $C_n$ by bit strings of length $O(\log n)$. Then the circuit family $(C_n)_{n \geq 0}$
is called \emph{\textsf{DLOGTIME}-uniform}  if (i) there exists a deterministic Turing machine that computes for a given gate $u \in \{0,1\}^*$
of $C_n$ ($|u| \in O(\log n)$) in time $O(\log n)$ the type of gate $u$, where the types are $x_1, \ldots, x_n$, \textsf{NOT}, \textsf{AND}, \textsf{OR}, or \textsf{Majority} gates,
and (ii) there exists a deterministic Turing machine that decides for two given gates $u,v \in \{0,1\}^*$ of $C_n$ ($|u|, |v| \in O(\log n)$) and a binary encoded integer $i$ with $O(\log n)$ many bits in time $O(\log n)$ whether $u$ is the $i$-th input gate for $v$.

\begin{definition}
Fix $k \geq 0$. We say that a language $L$ belongs to (uniform) $\textsf{NC}^{k}$ if there exist a (uniform) family of circuits $(C_{n})_{n \in \mathbb{N}}$ over the $\textsf{AND}, \textsf{OR}, \textsf{NOT}$ gates such that the following hold:
\begin{itemize}
\item The $\textsf{AND}$ and $\textsf{OR}$ gates take exactly $2$ inputs. That is, they have fan-in $2$.
\item $C_{n}$ has depth $O(\log^{k} n)$ and uses (has size) $n^{O(1)}$ gates. Here, the implicit constants in the circuit depth and size depend only on $L$.

\item $x \in L$ if and only if $C_{|x|}(x) = 1$. 
\end{itemize}
\end{definition}

\noindent The complexity class $\AC^{k}$ is defined analogously as $\textsf{NC}^{k}$, except that the $\textsf{AND}, \textsf{OR}$ gates are permitted to have unbounded fan-in.
That is, a single $\textsf{AND}$ gate can compute an arbitrary conjunction, and a single $\textsf{OR}$ gate can compute an arbitrary disjunction. The complexity class $\textsf{TC}^{k}$ is defined analogously as $\AC^{k}$, except that our circuits are now also permitted $\textsf{Majority}$ gates of unbounded fan-in.
We also allow circuits to compute functions by using multiple output gates. 

For every $k$, the following containments are well-known:
\[
\textsf{NC}^{k} \subseteq  \AC^{k} \subseteq \textsf{TC}^{k} \subseteq \textsf{NC}^{k+1}.
\]

\noindent In the case of $k = 0$, we have that:
\[
\textsf{NC}^{0} \subsetneq \AC^{0} \subsetneq \textsf{TC}^{0} \subseteq \textsf{NC}^{1} \subseteq \LogSpace \subseteq \textsf{NL} \subseteq \AC^{1}.
\]

\noindent We note that functions that are $\textsf{NC}^{0}$-computable can only depend on a bounded number of input bits. Thus, $\textsf{NC}^{0}$ is unable to compute the $\textsf{AND}$ function. It is a classical result that $\AC^{0}$ is unable to compute \algprobm{Parity} \cite{FSS}. The containment $\textsf{TC}^{0} \subseteq \textsf{NC}^{1}$ (and hence, $\textsf{TC}^{k} \subseteq \textsf{NC}^{k+1}$) follows from the fact that $\textsf{NC}^{1}$ can simulate the unbounded fan-in \textsf{Majority} gate.

The complexity class $\textsf{FOLL}$ is the set of languages decidable by uniform $\textsf{AC}$ circuits of depth $O(\log \log n)$ and polynomial-size \cite{BKLM}. It is known that $\textsf{AC}^{0} \subsetneq \FOLL \subsetneq \textsf{AC}^{1}$, the former by a simple diagonalization argument on top of Sipser's result \cite{SipserBorel}, and the latter because the \textsf{Parity} function is in $\mathsf{AC}^1$ but not $\FOLL$ (nor any depth $o(\log n / \log \log n)$). \FOLL\ cannot contain any complexity class that can compute \algprobm{Parity}, such as $\mathsf{TC}^0, \mathsf{NC}^1, \mathsf{L}$, or $\mathsf{NL}$, and it remains open whether any of these classes contain \FOLL.

Denote by $\textsf{FL}$ the class of logspace computable functions.

\subsection{Weisfeiler--Leman} \label{sec:WL}

We begin by recalling the Weisfeiler--Leman algorithm for graphs, which computes an isomorphism-invariant coloring. Let $\Gamma$ be a graph, and let $k \geq 2$ be an integer. The $k$-dimensional Weisfeiler--Leman, or $k$-WL, algorithm begins by constructing an initial coloring $\chi_{k,0} : V(\Gamma)^{k} \to \mathcal{K}$, where $\mathcal{K}$ is our set of colors, by assigning each $k$-tuple of vertices a color based on its isomorphism type. That is, two $k$-tuples $(v_{1}, \ldots, v_{k})$ and $(u_{1}, \ldots, u_{k})$ receive the same color under $\chi_{k,0}$ iff the map $v_i \mapsto u_i$ (for all $i \in [k]$) is an isomorphism of the induced subgraphs $\Gamma[\{ v_{1}, \ldots, v_{k}\}]$ and $\Gamma[\{u_{1}, \ldots, u_{k}\}]$ and for all $i, j$, $v_i = v_j \Leftrightarrow u_i = u_j$. 

For $r \geq 0$, the coloring computed at the $r$th iteration of  Weisfeiler--Leman is refined as follows. For a $k$-tuple $\overline{v} = (v_{1}, \ldots, v_{k})$ and a vertex $x \in V(\Gamma)$, define
\[
\overline{v}(v_{i}/x) = (v_{1}, \ldots, v_{i-1}, x, v_{i+1}, \ldots, v_{k}).
\]

The coloring computed at the $(r+1)$st iteration, denoted $\chi_{k, r+1}$, stores the color of the given $k$-tuple $\overline{v}$ at the $r$th iteration, as well as the colors under $\chi_{k,r}$ of the $k$-tuples obtained by substituting a single vertex in $\overline{v}$ for another vertex $x$. We examine this multiset of colors over all such vertices $x$. This is formalized as follows:
\begin{align*}
\chi_{k,r+1}(\overline{v}) = &( \chi_{r}(\overline{v}), \{\!\!\{ ( \chi_{r}(\overline{v}(v_{1}/x)), \ldots, \chi_{r}(\overline{v}(v_{k}/x) )) \bigr| x \in V(\Gamma) \}\!\!\} ),
\end{align*}
where $\{\!\!\{ \cdot \}\!\!\}$ denotes a multiset.

Note that the coloring $\chi_{k,r}$ computed at iteration $r$ induces a partition of $V(\Gamma)^{k}$ into color classes. The Weisfeiler--Leman algorithm terminates when this partition is not refined, that is, when the partition induced by $\chi_{k,r+1}$ is identical to that induced by $\chi_{k,r}$. The final coloring is referred to as the \textit{stable coloring}, which we denote $\chi_{k,\infty} := \chi_{k,r}$. 

The \textit{count-free} variant of $k$-WL works identically as the classical variant, except at the refinement step, we consider the set of colors rather than the  multi-set. We re-use the notation $\chi_{k,r}$ to denote the coloring computed by count-free $(k,r)$-WL; context should make it clear whether $\chi_{k,r}$ refers to count-free or counting WL (we never use $\chi_{k,r}$ to denote the count-free coloring when discussing counting WL, nor vice versa). 

Precisely:
\begin{align*}
\chi_{k,r+1}(\overline{v}) = &( \chi_{r}(\overline{v}), \{ ( \chi_{r}(\overline{v}(v_{1}/x)), \ldots, \chi_{r}(\overline{v}(v_{k}/x) )) \bigr| x \in V(\Gamma) \} ).
\end{align*}

Let $k \geq 2, r \geq 0$, and let $G$ be a graph. We say that the classical counting variant of $(k,r)$-WL \textit{distinguishes} $G$ from the graph $H$ if there exists a color class $C$ such that:
\[
|\{ \bar{x} \in V(G)^k : \chi_{k,r}(\bar{x}) = C \}| \neq |\{ \bar{x} \in V(H)^{k} : \chi_{k,r}(\bar{x}) = C\}|.
\]
Similarly, the count-free variant of $(k,r)$-WL \textit{distinguishes} $G$ from the graph $H$ if there exists a color class $C$ and some $\bar{x} \in V(G)^k$ where $\chi_{k,r}(\bar{x}) = C$, but for all $\bar{y} \in V(H)^k$, $\chi_{k,r}(\bar{y}) \neq C$. We say that the classical counting (resp., count-free) $(k,r)$-WL \textit{identifies} $G$ if for all $H \not \cong G$, $(k,r)$-WL (resp., count-free $(k,r)$-WL) distinguishes $G$ from $H$. 

The terms \emph{distinguish} and \emph{identify} also extend in the natural way when controlling only for the dimension $k$. The (count-free) \emph{Weisfeiler--Leman dimension} (\emph{WL-dimension}) of a graph $G$ is the minimum $k$, such that (count-free) $k$-WL identifies $G$.

Brachter and Schweitzer \cite{WLGroups} introduced three variants of WL for groups. We will restrict attention to the first two variants. WL Versions I and II are both executed directly on the Cayley tables, where $k$-tuples of group elements are initially colored. For WL Version I, two $k$-tuples $(g_{1}, \ldots, g_{k})$ and $(h_{1}, \ldots, h_{k})$ receive the same initial color iff (a) for all $i, j, \ell \in [k]$, $g_{i}g_{j} = g_{\ell} \iff h_{i}h_{j} = h_{\ell}$, and (b) for all $i, j \in [k]$, $g_{i} = g_{j} \iff h_{i} = h_{j}$. For WL Version II, $(g_{1}, \ldots, g_{k})$ and $(h_{1}, \ldots, h_{k})$ receive the same initial color iff the map $g_{i} \mapsto h_{i}$ for all $i \in [k]$ extends to an isomorphism of the generated subgroups $\langle g_{1}, \ldots, g_{k} \rangle$ and $\langle h_{1}, \ldots, h_{k} \rangle$. For both WL Versions I and II, refinement is performed in the classical manner as for graphs. Namely, for a given $k$-tuple $\overline{g}$ of group elements,
\begin{align*}
\chi_{k,r+1}(\overline{g}) = &( \chi_{r}(\overline{g}), \{\!\!\{ ( \chi_{r}(\overline{g}(g_{1}/x)), \ldots, \chi_{r}(\overline{g}(g_{k}/x) ) \bigr| x \in G \}\!\!\} ).
\end{align*}

The count-free variants of WL Versions I and II are defined in the identical manner as for graphs. The notions of whether (count-free) WL distinguishes two groups or identifies a group, are defined analogously as in the case for graphs.

\subsection{Pebbling Games} \label{sec:pebble}

\paragraph{Bijective Pebbling Game for Groups.} We recall the bijective pebble games for WL Versions I and II on groups \cite{Hella1989, Hella1993, WLGroups}. These games are often used to show that two groups $G$ and $H$ can(not) be distinguished by $k$-WL. The game is an Ehrenfeucht--Fra\"iss\'e game (cf., \cite{Ebbinghaus:1994, Libkin}), with two players: Spoiler and Duplicator. We begin with $k+1$ pairs of pebbles. Prior to the start of the game, each pebble pair $(p_{i}, p_{i}')$ is initially placed either beside the groups or on a given pair of group elements $v_{i} \mapsto v_{i}'$ (where $v_{i} \in G, v_{i}' \in H$). Each round $r$ proceeds as follows.
\begin{enumerate}
\item Spoiler picks up a pair of pebbles $(p_{i}, p_{i}^{\prime})$. 
\item We check the winning condition, which will be formalized later.
\item Duplicator chooses a bijection $f_{r} : G \to H$.
\item Spoiler places $p_{i}$ on some vertex $g \in G$. Then $p_{i}^{\prime}$ is placed on $f_{r}(g)$. 
\end{enumerate} 

Suppose that $(g_{1}, \ldots, g_{\ell}) \mapsto (h_{1}, \ldots, h_{\ell})$ have been pebbled. In Version I, Duplicator wins at the given round if this map satisfies the initial coloring condition of WL Version I: (a) for all $i, j, m \in [\ell]$, $g_{i}g_{j} = g_{m} \iff h_{i}h_{j} = h_{m}$, and (b) for all $i, j \in [\ell]$, $g_{i} = g_{j} \iff h_{i} = h_{j}$. In Version II, Duplicator wins at the given round if the map $(g_{1}, \ldots, g_{\ell}) \mapsto (h_{1}, \ldots, h_{\ell})$ extends to an isomorphism of the generated subgroups $\langle g_{1}, \ldots, g_{\ell} \rangle$ and $\langle h_{1}, \ldots, h_{\ell} \rangle.$ Brachter \& Schweitzer established that for $J \in \{I, II\}$, two $k$-tuples $\bar{g} \in G^k, \bar{h} \in H^k$ receive different colors under $(k,r)$-WL Version $J$ if and only if Spoiler has a winning strategy in the $(k+1)$-pebble, $r$-round pebble game starting from the initial configuration $\bar{g} \mapsto \bar{h}$ \cite{WLGroups}.

\paragraph{Count-Free Pebbling Game for Groups.} There exist analogous pebble games for count-free WL Versions I-II. The count-free $(k+1)$-pebble game consists of two players: Spoiler and Duplicator, as well as $(k+1)$ pebble pairs $(p, p^{\prime})$. In both versions, Spoiler wishes to show that the two groups $G$ and $H$ are not isomorphic. Duplicator wishes to show that the two groups are isomorphic. Prior to the start of the game, each pebble pair $(p_{i}, p_{i}')$ is initially placed either beside the groups or on a given pair of group elements $v_{i} \mapsto v_{i}'$ (where $v_{i} \in G, v_{i}' \in H$). Each round of the game proceeds as follows.
\begin{enumerate}
\item Spoiler picks up a pebble pair $(p_{i}, p_{i}^{\prime})$.
\item The winning condition is checked. The Version $J \in \{I,II\}$ winning condition here is the same as for the Version $J$ counting pebble game. 

\item Spoiler places one of the pebbles on some group element (either $p_{i}$ on some element of $G$ or $p_{i}'$ on some element of $H$). 
\item Duplicator places the other pebble on some element of the other group.
\end{enumerate}

Spoiler wins, by definition, at round $0$ if $G$ and $H$ do not have the same number of elements.  Suppose that $|G| = |H|$, and let $\bar{g} \in G^k, \bar{h} \in H^k$. We have that for $J \in \{I, II\}$, $\bar{g}$ and $\bar{h}$ receive different colors under the count-free $(k,r)$-WL Version $J$ if and only if Spoiler has a winning strategy in the $(k+1)$-pebble, $r$-round count-free pebble game starting from the initial configuration $\bar{g} \mapsto \bar{h}$ \cite{ImmermanLander1990, CFI, GrochowLevetWL}.

\subsection{Weisfeiler--Leman as a Parallel Algorithm} \label{sec:WLparallel}

Grohe \& Verbitsky \cite{GroheVerbitsky} previously showed that for fixed $k$, the classical $k$-dimensional Weisfeiler--Leman algorithm for graphs can be effectively parallelized. Precisely, each iteration of the classical counting WL algorithm (including the initial coloring) can be implemented using a logspace uniform $\textsf{TC}^{0}$ circuit, and each iteration of the \textit{count-free} WL algorithm can be implemented using a logspace uniform $\textsf{AC}^{0}$ circuit. A careful analysis shows that this parallelization holds under $\textsf{DLOGTIME}$-uniformity. As they mention (\cite[Remark~3.4]{GroheVerbitsky}), their implementation works for any first-order structure, including groups. In particular, the initial coloring of count-free WL Version I is $\textsf{DLOGTIME}$-uniform $\textsf{AC}^{0}$-computable, and each refinement step is $\textsf{DLOGTIME}$-uniform $\textsf{AC}^{0}$-computable. However, because for groups we have different versions of WL, we explicitly list out the resulting parallel complexities, which differ slightly between the versions.

\begin{itemize}
\item \textbf{WL Version I:} Let $(g_{1}, \ldots, g_{k})$ and $(h_{1}, \ldots, h_{k})$ be two $k$-tuples of group elements. We may test in $\textsf{AC}^{0}$ whether (a) for all $i, j, m \in [k]$, $g_{i}g_{j} = g_{m} \iff h_{i}h_{j} = h_{m}$, and (b) $g_{i} = g_{j} \iff h_{i} = h_{j}$. So we may decide if two $k$-tuples receive the same initial color in $\textsf{AC}^{0}$. Comparing the (multi)set of colors at the end of each iteration (including after the initial coloring), as well as the refinement steps, proceeds identically as in \cite{GroheVerbitsky}. Thus, for fixed $k$, each iteration of the counting $k$-WL Version I can be implemented using a $\textsf{DLOGTIME}$-uniform $\textsf{TC}^{0}$ circuit, and each iteration of the count-free $k$-WL Version I can be implemented using a $\textsf{DLOGTIME}$-uniform $\textsf{AC}^{0}$ circuit. 

\item \textbf{WL Version II:} Let $(g_{1}, \ldots, g_{k})$ and $(h_{1}, \ldots, h_{k})$ be two $k$-tuples of group elements. We may use the marked isomorphism test of Tang \cite{TangThesis} to test in $\textsf{L}$ whether the map sending $g_{i} \mapsto h_{i}$ for all $i \in [k]$ extends to an isomorphism of the generated subgroups $\langle g_{1}, \ldots, g_{k} \rangle$ and $\langle h_{1}, \ldots, h_{k} \rangle$. So we may decide whether two $k$-tuples receive the same initial color in $\textsf{L}$. Comparing the (multi)set of colors at the end of each iteration (including after the initial coloring), as well as the refinement steps, proceed identically as in \cite{GroheVerbitsky}. Thus, for fixed $k$, the initial coloring of both the counting and count-free $k$-WL Version II is $\textsf{L}$-computable. Each refinement step of the counting $k$-WL Version II is computable using a $\textsf{DLOGTIME}$-uniform $\textsf{TC}^{0}$-circuit, and each refinement step of the count-free $k$-WL Version II is computable using a $\textsf{DLOGTIME}$-uniform $\textsf{AC}^{0}$-circuit.
\end{itemize}

\section{Weisfeiler--Leman and Direct Products}

In this section, we will establish Theorem~\ref{thm:MainWL}.

\subsection{Additional Preliminaries for Direct Products} \label{sec:DirectProductPreliminaries}

We will begin by recalling additional preliminaries for direct products. 

\begin{definition}
We say that a group $G$ is \textit{directly indecomposable} if the only direct product decomposition of $G$ is $\{G\}$.     
\end{definition}

\begin{definition}
Let $\mathcal{S} := \{S_1, \ldots, S_k\}$ be a direct product decomposition of the group $G$. We say that $\mathcal{S}$ is \textit{fully refined} if for each $i \in [k]$, $S_{i}$ is indecomposable.    
\end{definition}

The Remak--Krull--Schmidt theorem guarantees that for any group $G$ and any two fully refined direct product decompositions $\mathcal{S}, \mathcal{T}$ of $G$, the isomorphism types of the indecomposable factors and their multiplicities are the same in $\mathcal{S}$ as in $\mathcal{T}$. However, in general, $\mathcal{S} \neq \mathcal{T}$. Consider, for instance, the Klein group $\mathbb{V}_{4} \cong (\mathbb{Z}/2\mathbb{Z})^{2}$. Let $E_1 = \langle (1,0) \rangle$, $E_2 = \langle (0, 1) \rangle$ and $F = \langle (1,1) \rangle $. Now, $\mathbb{V}_{4} = E_1 \times E_2 = E_1 \times F$, but $E_2 \neq F$, even though $E_2 \cong F$.

We will now recall results on when any two fully refined direct product decompositions $\mathcal{S}, \mathcal{T}$ of $G$ satisfy $\mathcal{S} = \mathcal{T}$, as well as when a given subgroup $N \trianglelefteq G$ appears (as a set) in any fully refined direct product decomposition of $G$. 

\begin{definition}
Let $G$ be a group. A \textit{central automorphism} $\alpha$ of $G$ is an automorphism of $G$ that acts trivially on $G/Z(G)$, that is, $\alpha(gZ(G)) = gZ(G)$ for every $g\in G$.    
\end{definition}

\begin{theorem}[{\cite[Theorem~3.3.8]{Robinson1982}}] \label{thm:CentralAutomorphisms}
Let $G$ be a group, and let $\mathcal{S} = \{S_1, \ldots, S_k\}$, $\mathcal{T} = \{T_1, \ldots, T_k\}$ be fully refined direct product decompositions of $G$. Then there exists a central automorphism $\alpha$ of $G$ and a permutation $\pi \in \text{Sym}(k)$ such that $\alpha(S_i) = T_{\pi(i)}$ for every $i\in[k]$.
\end{theorem}

We use Theorem~\ref{thm:CentralAutomorphisms} to establish conditions for when a group $G$ admits a unique fully refined direct product decomposition.

\begin{theorem} \label{thm:UniqueDecomposition}
Let $G$ be a group, and let $\mathcal{S} = \{S_1, \ldots, S_k\}$ be a fully refined direct product decomposition of $G$. For any $i \in [k]$, we have that $S_i$ appears in any fully refined direct product decomposition of $G$ if and only if for all distinct $j \neq i$, $|S_{i}/[S_{i},S_{i}]|$ and $|Z(S_j)|$ are coprime. In particular, $\mathcal{S}$ is unique if and only if for all distinct $i, j \in [k]$, $|S_{i}/[S_{i},S_{i}]|$ and $|Z(S_j)|$ are coprime.    
\end{theorem}

\begin{proof}
For direct product decompositions $A = \prod_{u=1}^k A_u$ and $B=\prod_{v=1}^\ell B_v$ let us write a homomorphism $\gamma:A\to B$ as $\prod_{u=1}^k\gamma_u$, where $\gamma_u$ is the restriction of $\gamma$ to $A_u$, so that $\gamma(a_1\cdots a_k) = \gamma(a_1)\cdots\gamma(a_k) = \gamma_1(a_1)\cdots\gamma_k(a_k)$, where $(a_1,\dots,a_k)\in A$. Further, let $\gamma_u = \prod_{v=1}^\ell \gamma_{uv}$, where $\gamma_{uv}:A_u\to B_v$.

Throughout the proof, fix $i\in[k]$. Suppose that $|S_i/[S_i,S_i]|$ and $|Z(S_j)|$ are coprime for every $j\ne i$. Let $\mathcal T=\{T_1,\dots,T_k\}$ be another fully refined product decomposition of $G$. By Theorem~\ref{thm:CentralAutomorphisms}, there exists a central automorphism $\alpha$ and a permutation $\pi \in \text{Sym}(k)$ such that $\alpha(S_i) = T_{\pi(i)}$. We will show that $\alpha(S_i)=S_i$, witnessing that $S_i$ appears in $\mathcal T$.

Since $\alpha$ is central, for every $g\in G$ there is $\beta(g)\in Z(G)$ such that $\alpha(g)=g\beta(g)$. Since $\alpha$ is a homomorphism, we have $gh\beta(gh) = \alpha(gh)=\alpha(g)\alpha(h)=g\beta(g)h\beta(h) = gh\beta(g)\beta(h)$, so $\beta:G\to Z(G)$ is a homomorphism. Since $G=\prod_{u=1}^k S_u$, we have $Z(G) = \prod_{u=1}^k Z(S_u)$. Let us write $\beta$ as $\beta=\prod_{u=1}^k\beta_u$ and $\beta_u=\prod_{v=1}^k\beta_{uv}$, where $\beta_{uv}:S_u\to Z(S_v)$. Since $|S_i/[S_i,S_i]|$ and $|Z(S_j)|$ are coprime for every $j\ne i$, the homomorphisms $\beta_{ij}$ are trivial for every $j\ne i$. Then $\alpha(S_i) \subseteq S_i\beta(S_i) = S_i\beta_i(S_i) = S_i\beta_{ii}(S_i)\subseteq S_iZ(S_i)\subseteq S_i$. Since $\alpha$ is an automorphism, $\alpha(S_i)=S_i$. 

Conversely, suppose that there is a $j\ne i$ such that $|S_i/[S_i,S_i]|$ and $|Z(S_j)|$ are not coprime. We will construct a fully refined direct product decomposition of $G$ in which $S_i$ does not appear. Let $\beta_{ij}:S_i\to Z(S_j)$ be any nontrivial homomorphism. Set $\beta_{iv}=1$ for every $v\ne j$, $\beta_i=\prod_{v=1}^k\beta_{iv}$, $\beta_u=1$ for every $u\ne i$, and $\beta = \prod_{u=1}^k\beta_u$. Then $\beta$ is a homomorphism $G\to Z(G)$. Let $\alpha:G\to G$ be given by $\alpha(g)=g\beta(g)$. Since $\beta(G)\subseteq Z(G)$, $\alpha$ is a central homomorphism. We will show that $\alpha$ is an automorphism. It suffices to show that $\alpha$ is injective. Let $g=(g_1,\dots,g_k)\in G = \prod_{u=1}^k S_u$. Since $\beta_{ij}$ is the only nontrivial component of $\beta$, we have $\alpha(g_1,\dots,g_k) = (g_1,\dots,g_{j-1},g_j',g_{j+1},\dots,g_k)$, where $g_j' = g_j\beta_{ij}(g_i)$. Suppose that $h=(h_1,\dots,h_k)\in G$ and $\alpha(g)=\alpha(h)$. Then $g_v=h_v$ for every $v\ne j$, and $g_j\beta_{ij}(g_i) = h_j\beta_{ij}(h_i)$. As $g_i=h_i$, we conclude that $g_j=h_j$, too, and $g=h$.

For every $u\in[k]$ let $T_u=\alpha(S_u)$. Since $\alpha$ is an automorphism, $\mathcal T=\{T_1,\dots,T_k\}$ is a fully refined direct product decomposition of $G$. We will show that $S_i\not\in\mathcal T$. We have $T_u=S_u$ for every $u\ne i$, so it suffices to show that $T_i\ne S_i$. Let $g_i\in S_i$ be such that $\beta_{ij}(g_i)\ne 1$. Since $\beta_{ij}(S_i)\subseteq S_j$ and $S_i\cap S_j=1$, it follows that $\beta_{ij}(g_i)\not\in S_i$. Now, $T_i=\alpha(S_i)$ contains $\alpha(g_i)=g_i\beta(g_i)=g_i\beta_i(g_i) = g_i\beta_{ij}(g_i)\not\in S_i$, and hence $T_i\ne S_i$.
\end{proof}

We obtain the following corollary.

\begin{corollary} \label{cor:PerfectCenterless}
Let $G$ be a group, and let $N \trianglelefteq G$ be directly indecomposable. If $N$ is perfect or centerless, then $N$ appears in any fully-refined direct product decomposition of $G$.    
\end{corollary}

\begin{proof}
If $N$ is perfect, then $N = [N,N]$, and so $|N/[N,N]| = 1$. If $N$ is centerless, then $|Z(N)| = 1$. The result now follows by Theorem~\ref{thm:UniqueDecomposition}.    
\end{proof}

\subsection{Weisfeiler--Leman and Direct Products}

We recall that $\mathcal{C}$ is the class of groups $G$ that admit a (unique) direct product decomposition, where each indecomposable direct factor is (i) $O(1)$-generated and (ii) either perfect or centerless (see Definition~\ref{def:PerfectCenterless}). Let $G \in \mathcal{C}$ be a group of order $n$. Our first step will be to show that if $H \not \in \mathcal{C}$, then count-free WL Version II will distinguish $H$ from in $O(\log \log n)$ rounds. We will accomplish this using the Rank Lemma due to Grochow and Levet \cite[Lemma~4.3]{GrochowLevetWL}. We will first introduce a notion of weight.

\begin{definition}[cf. {\cite[Definition~4.9]{GLDescriptiveComplexity}}] \label{def:weight}
Let $H \leq G$ be a characteristic subgroup with a unique fully refined direct product decomposition $\mathcal{S} = \{ S_1, \ldots, S_k\}$. Let $g \in G$. If $g \in H$, write $g = \prod_{i=1}^{k} s_{i}$, where $s_i \in S_i$. In this case, define the \textit{weight} of $g$, denoted $\wt(g) := |\{ i : s_i \neq 1 \}|$. Otherwise, we have that $g \not \in H$, in which case we define $\wt(g) := \infty$.
\end{definition}

Note that as $\mathcal{S}$ is the unique fully refined direct product decomposition of $H$, the decomposition of $g = \prod_{i=1}^{k} s_{i}$ is unique up to the ordering of the factors. In particular, the weight of an element depends only on $G$ (respectively, $H$), and not a choice of decomposition.

\begin{remark}
Definition~\ref{def:weight} generalizes \cite[Definition~4.9]{GLDescriptiveComplexity}, where Grochow and Levet considered $G$ to have no Abelian normal subgroups and $H$ to be the socle of $G$, denoted $\Soc(G)$. (Note that $\Soc(G)$ is the direct product of the minimal normal subgroups of $G$.) In this setting, the $S_i$ in the unique refined direct product decomposition of $\Soc(G)$ were the non-Abelian simple direct factors of $\Soc(G)$.    
\end{remark}

As in \cite{GLDescriptiveComplexity}, our notion of weight is a special case of \textit{rank}, as introduced in \cite{GrochowLevetWL}.

\begin{definition}[{\cite[Definition~4.1]{GrochowLevetWL}}]
Let $C \subseteq G$ be a subset of a group $G$ that is closed under taking inverses. We define the $C$-\textit{rank} of $g \in G$, denoted $\rk_{C}(g)$, as the minimum $m$ such that $g$ can be written as a word of length $m$ in the elements of $C$. If $g$ cannot be so written, we define $\rk_{C}(g) = \infty$.
\end{definition}

\begin{lemma}[{\cite[Rank Lemma~4.3]{GrochowLevetWL}}] \label{lem:RankLemma}
Let $C \subseteq G$ be a subset of $G$ that is closed under taking inverses. Let $J \in \{I,II\}$. Suppose that if $x \in C$ and $y \not \in C$, then Spoiler can win in the count-free Version $J$ pebble game from the initial configuration $x \mapsto y$, using $k$ pebbles and $r$ rounds. If $\rk_{C}(g) \neq \rk_{C}(h)$, then Spoiler can win in the count-free Version $J$ game with $k+1$ pebbles and $r + \log(d) + O(1)$ rounds, where $d = \text{diam}(\text{Cay}(\langle C \rangle, C)) \leq |\langle C \rangle| \leq |G|$.
\end{lemma}

We will begin with the following lemma, which shows that count-free WL can detect the weight one elements.

\begin{lemma} \label{lem:RankOne}
Let $G$ be a group of order $n$, and let $S \trianglelefteq G$ be $O(1)$-generated. Let $H$ be an arbitrary group. Let $g \in S$, $h \in H$. Suppose that there does not exist $T \trianglelefteq H$ such that $h \in T$ and $T \cong S$. Then in the Version II pebble game, Spoiler can win with $O(1)$ pebbles and $O(1)$  rounds, from the initial configuration $g \mapsto h$.  
\end{lemma}

\begin{proof}
Let $k := d(S)$. We consider the following cases.
\begin{itemize}
\item \textbf{Case 1:} Suppose first that there does not exist some $T \leq H$ such that $h \in T$ and $T \cong S$. Spoiler uses $k$ pebbles and $k$ rounds to pebble generators $(g_1, \ldots, g_k)$ for $S$. Let $(h_1, \ldots, h_k)$ be Duplicator's response. As there does not exist a copy of $S$ in $H$ that also contains $h$, it follows that $(g, g_1, \ldots, g_k) \mapsto (h, h_1, \ldots, h_k)$ does not extend to an isomorphism of $S = \langle g, g_1, \ldots, g_k \rangle$ and $\langle h, h_1, \ldots, h_k \rangle$. So Spoiler wins with $k$ pebbles and $k$ rounds.

\item \textbf{Case 2:} Now suppose instead that there exists some $T \leq H$ such that $h \in T$ and $T \cong S$. By assumption of the lemma, we have that $T \not \trianglelefteq H$. Spoiler begins by pebbling generators $(g_1, \ldots, g_k)$ for $S$. Let $(h_1, \ldots, h_k)$ be Duplicator's response. We may assume that the map $(g, g_1, \ldots, g_k) \mapsto (h, h_1, \ldots, h_k)$ extends to an isomorphism of $S = \langle g_1, \ldots, g_k \rangle$ and $T := \langle h_1, \ldots, h_k \rangle$ (and in particular, that $h \in T$). Otherwise, Spoiler wins. As $T \not \trianglelefteq H$, there exists some $x \in H$ such that $xTx^{-1} \neq T$. Spoiler pebbles such an $x \in H$, and Duplicator responds by pebbling some $y \in G$. As $S$ is normal in $G$, it follows that $ySy^{-1} = S$. Thus, the map $(g, g_1, \ldots, g_k, y) \mapsto (h, h_1, \ldots, h_k, x)$ does not extend to an isomorphism of $\langle g, g_1, \ldots, g_k, y \rangle$ and $\langle h, h_1, \ldots, h_k, x \rangle$. So Spoiler wins with $k+1$ pebbles and in $k+1$ rounds.  \qedhere

\end{itemize}

\end{proof}

\begin{lemma} \label{lem:CountFreeDecompose}
Let $G \in \mathcal{C}$ and $H \not \in \mathcal{C}$ be groups of order $n$. The count-free $(O(1), O(\log \log n))$-WL Version II will distinguish $G$ from $H$.  
\end{lemma}

\begin{proof}
Let $S \subseteq G$ consist of the set of elements $g$, such that for each such $g$, there exists some $O(1)$-generated $N_{g} \trianglelefteq G$, where (i) $N_{g}$ is either perfect or centerless and (ii) $g \in N_{g}$. We have by Lemma~\ref{lem:RankOne} that if $g \in S$ and $g' \not \in S$, then Spoiler can win with $O(1)$ pebbles and $O(1)$ rounds from the initial configuration $g \mapsto g'$. Thus, the hypotheses of the Rank~Lemma~\ref{lem:RankLemma} are satisfied.

As $H \not \in \mathcal{C}$, there exists some $h \in H$ such that $\wt(h) = \infty$. Spoiler begins by pebbling such an $h$. Let $g \in G$ be Duplicator's response. As $G \in \mathcal{C}$, $\wt(g) < \infty$. By Lemma~\ref{lem:RankLemma} and as $G$ has at most $O(\log n)$ direct factors, Spoiler now wins with $O(1)$ pebbles and $O(\log \log n)$ additional rounds. The result now follows.
\end{proof}

The proof of Theorem~\ref{thm:MainWL}(a) and Theorem~\ref{thm:MainSimple} are essentially identical. We will thus handle them together. In order to do so, we will recall some additional results concerning non-Abelian simple groups.

\begin{lemma}[{\cite{CollinsLevetWL}}] \label{lem:WLSimple}
Let $G$ be a direct product of non-Abelian simple groups. If $H$ is not a direct product of non-Abelian simple groups, then in the count-free Version I pebble game,  Spoiler can win using $O(1)$ pebbles and $O(\log \log n)$ rounds.
\end{lemma}

\begin{proof}
Let $G$ be a group of order $n$, such that $G$ decomposes as a direct product of non-Abelian simple groups. In particular, we write $G = \prod_{i=1}^{k} S_i$, where $S_i$ is non-Abelian simple ($i \in [k]$). Let $H$ be a group of order $n$ such that $H \not \cong G$. 

We first note that $G$ has no Abelian normal subgroups. If $H$ has an Abelian normal subgroup, we have by \cite[Lemma~6.6]{CollinsLevetWL} that Spoiler can win with $O(1)$ pebbles and $O(1)$ rounds. Now suppose that $H$ has no Abelian normal subgroups, but that $H$ does not decompose as a direct product of non-Abelian simple groups. By \cite[Lemma~6.12]{CollinsLevetWL}, Spoiler can win with $O(1)$ pebbles and $O(\log \log n)$ rounds.
\end{proof}

We will also take advantage of the following result due to Kassabov and Riley. 

\begin{theorem}[{\cite{KassabovRiley}}] \label{thm:BabaiKantorLubotsky}
There exists an absolute constant $C$, such that for every finite simple group $S$, there exists a generating set $\{g_1, \ldots, g_4 \} \subseteq S$ such that every element of $S$ can be realized as a word over $\{g_1, \ldots, g_4\}$ of length at most $C \log |S|$.
\end{theorem}

\begin{remark}
The Classification of Finite Simple Groups provides that every finite simple group is generated by at most two elements $\{a,b\}$. As the diameter of the Cayley graph $\text{Cay}(S, \{a,b\})$ is at most $|S|$, it follows that every element in $S$ can be realized as a word of length at most $|S|$ over $\{a,b\}$. We are not aware of better bounds for such $2$-element generating sets.
\end{remark}

\begin{proposition} \label{prop:MainA}
Let $G \in \mathcal{C}$ be a group of order $n$. 
\begin{enumerate}[label=(\alph*)]
\item (Theorem~\ref{thm:MainWL}(a)) The count-free $(O(1), O(\log \log n))$-WL Version II identifies $G$.

\item (Theorem~\ref{thm:MainSimple}) If furthermore, $G$ admits a decomposition as a direct product of non-Abelian simple groups, then the count-free $(O(1), O(\log \log n))$-WL Version I identifies $G$.
\end{enumerate}
\end{proposition}

\begin{proof}
\noindent 
\begin{itemize}
\item \textbf{Proof of Theorem~\ref{thm:MainWL}(a).} Let $H$ be a group of order $n$ such that $G \not \cong H$. By Lemma~\ref{lem:CountFreeDecompose}, we may assume that $H \in \mathcal{C}$; otherwise, Spoiler wins with $O(1)$ pebbles and $O(\log \log n)$ rounds.

Consider the fully refined direct product decompositions $G = \prod_{i=1}^{k} S_{i}$ and $H = \prod_{i=1}^{k'} T_{i}$. As $G, H \in \mathcal{C}$, we have for each $i \in [k]$ and each $j \in [k']$ that $S_{i}, T_{j}$ are $O(1)$-generated, and either perfect or centerless. 

We first claim that $k = k'$; or Spoiler wins with $O(1)$ pebbles and $O(\log k) \leq O(\log \log n)$ rounds. Without loss of generality, suppose that $k > k'$. Spoiler begins by pebbling some element $g \in G$ such that $\wt(g) = k$. Let $h \in H$ be Duplicator's response. As $k' < k$, we have that $\wt(h) < k$. Thus, by the Rank~Lemma~\ref{lem:RankLemma}, Spoiler can win with $O(1)$ additional pebbles and $O(\log k) \leq O(\log \log n)$ additional rounds.

We may now assume that $k = k'$. Spoiler begins by pebbling some $g \in G$ of weight $k$. Let $h \in H$ be Duplicator's response. We may again assume that $\wt(g) = \wt(h)$; otherwise by the Rank~Lemma~\ref{lem:RankLemma}, Spoiler wins with $O(1)$ pebbles and in $O(\log k) \leq O(\log \log n)$ rounds.

By the Remak--Krull--Schmidt Theorem, as $G \not \cong H$, we have for any permutation $\pi \in \text{Sym}(k)$, that there exists some $i \in [k]$ such that $S_i \not \cong T_{\pi(i)}$. Now we may write $g, h$ uniquely (up to reordering of the factors) as $g := \prod_{i=1}^{k} s_{i}$ and $h := \prod_{i=1}^{k} t_i$, where (without loss of generality) for each $i \in [k]$, $s_i \in S_i$ and $t_i \in T_i$.

Let:
\begin{align*}
&g' := \prod_{i=1}^{\lceil k/2 \rceil} s_i, \text{ and } \,
g'' := \prod_{i=\lceil k/2 \rceil + 1}^{k} s_i.    
\end{align*} 
Spoiler begins by pebbling $g', g''$. Let $h', h''$ be Duplicator's response. As $g = g' \cdot g''$, we have that $h = h' \cdot h''$, or Spoiler immediately wins. By the Rank Lemma~\ref{lem:RankLemma}, we have that $\wt(g') = \wt(h')$ and $\wt(g'') = \wt(h'')$, or Spoiler wins with $O(1)$ additional pebbles and $O(\log k) \leq O(\log \log n)$ additional rounds.

Let $\mathcal{S}' := \{ S_1, \ldots, S_{\lceil k/2 \rceil} \}$, and let $\mathcal{S}'' := \{ S_{\lceil k/2\rceil+1}, \ldots, S_{k}\}$. By relabeling the direct factors of $\mathcal{T}$, we let $\mathcal{T}' := \{ T_{1}, \ldots, T_{\lceil k/2 \rceil}\}$ such that for each $i \in \lceil [k/2 \rceil]$, there exists $t_i \in T_i$ such that $h' = t_1 \cdots t_{\lceil k/2 \rceil}$. Let $\mathcal{T}'' := \{ T_{\lceil k/2 \rceil+1}, \ldots, T_{k} \}$. As $G \not \cong H$, one of the following conditions necessarily holds:
\begin{itemize}
 \item There does not exist a bijection $\varphi : \mathcal{S}' \to \mathcal{T}'$ such that $S \cong \varphi(S)$ for all $S \in \mathcal{S}'$, or
 
 \item There does not exist a bijection $\varphi : \mathcal{S}'' \to \mathcal{T}''$ such that $S \cong \varphi(S)$ for all $S \in \mathcal{S}''$.
\end{itemize}

Without loss of generality, suppose that there does not exist a bijection $\varphi : \mathcal{S}' \to \mathcal{T}'$ such that $S \cong \varphi(S)$ for all $S \in \mathcal{S}'$.

Thus, Spoiler now reuses the pebbles on $g, g''$, and we iterate on this strategy starting from $g' \mapsto h'$. Thus, after $O(\log k) \leq O(\log \log n)$ rounds, we have reached the case where $\wt(g') = \wt(h') = 1$. By  Lemma~\ref{lem:RankOne}, Spoiler now wins with $O(1)$ additional pebbles and $O(1)$ additional rounds. In total, Spoiler used $O(1)$ pebbles and $O(\log \log n)$ rounds.

\item \textbf{Proof of Theorem~\ref{thm:MainSimple}.} Suppose now that $G$ admits a decomposition as a direct product of non-Abelian simple groups. By Lemma~\ref{lem:WLSimple}, we may assume that $H$ also admits a decomposition as a direct product of non-Abelian simple groups. Otherwise, Spoiler wins with $O(1)$ pebbles and $O(\log \log n)$ rounds. The proof now proceeds identically as in Theorem~\ref{thm:MainWL}(a), until we reach the case where $g' \mapsto h'$ have been pebbled and satisfy $\wt(g') = \wt(h') = 1$. In order to win, Spoiler pebbles the $4$-element generating sequence (prescribed by Theorem~\ref{thm:BabaiKantorLubotsky}) for the non-Abelian simple factor $S \trianglelefteq G$ containing $g'$. By \cite[Section~3]{CollinsLevetWL} Spoiler now wins with $O(1)$ additional pebbles and $O(\log \log n)$ additional rounds. The result now follows. \qedhere
\end{itemize}
\end{proof}

\begin{proposition} \label{prop:MainB}
Let $G \in \mathcal{C}$. The counting $(O(1), O(1))$-WL Version II identifies $G$.
\end{proposition}

\begin{proof}
Let $H$ be a group of order $n:= |G|$ such that $H \not \cong G$. Duplicator chooses a bijection $f : G \to H$. As $G \not \cong H$, there exists some $g \in G$ such that $\wt(g) = 1$, and either: (i) $\wt(f(g)) > 1$, or (ii) $\wt(g) = 1$ and the (unique) direct factor $S \trianglelefteq G$ containing $g$ is not isomorphic to the (unique) direct factor $T \trianglelefteq H$ containing $f(g)$. Both of these cases are handled by  Lemma~\ref{lem:RankOne}, which provides that Spoiler wins with $O(1)$ pebbles and $O(1)$ rounds. The result now follows.
\end{proof}

We now prove Theorem~\ref{thm:MainWL}.

\begin{proof}[Proof of Theorem~\ref{thm:MainWL}]
Part (a) was established in Proposition~\ref{prop:MainA}(a), and part (b) was established in Proposition~\ref{prop:MainB}.
\end{proof}

\section{Direct Product Decompositions in Parallel}

In this section, we will prove Theorem~\ref{thm:MainDecompose}. We accomplish this by parallelizing the work of Kayal and Nezhmetdinov \cite{KayalNezhmetdinov}.

\subsection{Group Division in Parallel}

In this section, we consider the \algprobm{Group Division} problem, which takes as input a group $G$ given by its multiplication table and $A \trianglelefteq G$ as a set of elements, and asks for a $B \trianglelefteq G$ such that $G = A \times B$, if such a $B$ exists. Following the strategy of Kayal and Nezhmetdinov \cite[Section~4]{KayalNezhmetdinov}, we will break up the problem into the following two cases: when $G/A$ is Abelian, and when $G/A$ is non-Abelian.

We will consider first the case when $G/A$ is Abelian. We will first some key preliminaries.

\begin{definition}
Let $G$ be a group. We say that $x \in G$ \textit{splits} from $G$ if there exists $H \leq G$ such that $G = H \times \langle x \rangle$.    
\end{definition}

The following lemma characterizes when an element splits from an Abelian $p$-group.

\begin{lemma}[{\cite[Lemma~6.6]{BrachterSchweitzerWLLibrary}}] \label{lem:Split}
Let $G$ be an Abelian $p$-group. An element $x \in G$ splits from $G$ if and only if there does not exist $y \in G$ such that $|xy^p| < |x|$.
\end{lemma}

\begin{lemma}[{\cite[Proposition~3.1]{BKLM}}] \label{lem:Order}
Let $G$ be a group, and let $x \in G$. We can compute $|x|$, as well as $\langle x \rangle$, in $\textsf{FL} \cap \textsf{FOLL}$.
\end{lemma}

\begin{remark}
The complexity of order finding was recently improved to $\forall^{\log n}\DTISPpll \subseteq \textsf{FL} \cap \FOLL$ \cite{CGLWISSAC}.
\end{remark}

\begin{proposition} \label{prop:DecomposeAbelianGroup}
Let $G$ be an Abelian group. We can compute a basis, and hence a fully-refined direct product decomposition, for $G$ in $\textsf{AC}^{1}$.    
\end{proposition}

It is possible to compute a basis for an Abelian group in $\textsf{NC}^{2}$ using Smith normal form \cite{VillardSNF}. Proposition~\ref{prop:DecomposeAbelianGroup} improves this bound to $\textsf{AC}^{1}$, when the underlying group is given by its multiplication table.

\begin{proof}
We first decompose $G$ into a direct product of its Sylow subgroups. As $G$ is Abelian, we have that for each prime $p \mid |G|$, there exists exactly one Sylow $p$-subgroup $A_{p}$ of $G$. We may, in $\textsf{FL} \cap \textsf{FOLL}$, compute the Sylow subgroups of $G$ \cite[Lemma~3.1]{BKLM}. It remains to directly decompose each Sylow subgroup of $G$, which we will do in parallel.

Fix a prime $p$, and consider the Sylow $p$-subgroup $A_{p}$ of $G$. We proceed as follows. First, for each $g \in A_{p}$, we test whether $g$ splits from $A_{p}$. For some fixed $g$, we may use Lemma~\ref{lem:Order} to decide in $\textsf{FL} \cap \textsf{FOLL}$ whether there exists some $h \in A_{p}$ such that $|gh^{p}| < |g|$. By Lemma~\ref{lem:Split}, we have that $g$ splits from $A_{p}$ precisely if no such $h$ exists. Thus, in $\textsf{FL} \cap \textsf{FOLL}$, we can identify the elements of $A_{p}$ that split from $A_{p}$. Next, using Lemma~\ref{lem:Order}, we may for each $g \in A_{p}$, write down the elements of $\langle g \rangle$ in $\textsf{FL} \cap \textsf{FOLL}$.

We now turn to constructing a fully-refined direct product decomposition of $A_{p}$. We iteratively construct a generating set in the following manner. Let $g_{1} \in A_{p}$ be such that $g_{1}$ splits from $A_{p}$. As we have already identified the elements that split from $A_{p}$ in the preceding paragraph, we may select such a $g_{1}$ in $\textsf{AC}^{0}$. Let $S_{1} = \langle g_1 \rangle$. As we have already written down $\langle g_1 \rangle$ in the preceding paragraph, we may identify the elements of $S_1$ in $\textsf{AC}^{0}$.

Now fix $k \geq 1$. Let $H_{k} = \{ g_1, \ldots, g_k \}$, and let $S_{k} = \langle H_{k} \rangle$. Suppose that we have written down all of the elements in $S_{k}$. Furthermore, suppose that for all $i \in [k]$, $g_{i}$ splits from $A_{p}$ and $g_{i} \not \in \langle g_1, \ldots, g_{i-1} \rangle$. Finally, suppose that we have constructed $S_{k}$ using an $\textsf{AC}$ circuit of depth $O(k)$ and size $\poly(n)$. We now have the following cases. We may first test whether $S_{k} = A_{p}$ in $\textsf{AC}^{0}$. In this case, $\{ \langle g_1 \rangle, \ldots, \langle g_k \rangle \}$ forms a basis of $A_p$.

Suppose instead that $S_{k} \neq A_{p}$. Then there exists some $g_{k+1} \in A_{p} \setminus S_{k}$. As $A_{p}$ is Abelian, $A_{p}$ admits a decomposition into a direct product of cyclic groups. In particular, at least one such $g_{k+1} \in A_{p} \setminus S_{k}$ splits from $A_{p}$. Again, as we have previously identified the elements that split from $A_{p}$ as well as the elements of $S_{k}$, we may, in $\textsf{AC}^{0}$, select some $g_{k+1} \in A_{p} \setminus S_{k}$ that splits from $A_{p}$. Let $H_{k+1} = H_{k} \cup \{g_{k+1}\}$. We will now show how to write down 
\[
S_{k+1} := \langle H_{k+1} \rangle = S_{k+1} \cdot \langle g_{k+1} \rangle = \{ xy : x \in S_{k+1}, y \in \langle g_{k+1} \rangle\}
\]
in $\textsf{AC}^{0}$. By the inductive hypothesis, we have already explicitly written down the elements of $S_{k}$. Furthermore, in the second paragraph, we have already written down the elements of $\langle g_{k+1} \rangle$. Thus, we may compute $S_{k+1} \cdot \langle g_{k+1} \rangle$ in $\textsf{AC}^{0}$, as desired.

Note that only $\log |A_{p}| \leq \log |G|$ many iterations are required to compute a fully-refined direct product decomposition of $A_{p}$. Our pre-processing steps in the first two paragraphs are $\textsf{FL}$-computable, and the remaining work is $\textsf{AC}^{1}$-computable. Thus, we may compute a fully-refined direct product decomposition of $A_{p}$ in $\textsf{AC}^{1}$. The result now follows.
\end{proof}

\begin{lemma} \label{lem:AbelianDivision}
Let $G$ be a finite group, and let $A \trianglelefteq G$ such that $G/A$ is Abelian, with direct product decomposition $\prod_{i=1}^{k} \langle g_{i}A \rangle$. Then there exists $B \trianglelefteq G$ such that $G = A \times B$ if and only if for each $i \in [k]$, there exists $b_i \in Z(G) \cap g_{i}A$ satisfying $|b_{i}| = |g_{i}A|$.    
\end{lemma}

\begin{proof}
Suppose first that such a $B$ exists. Then every coset $g_iA$ is of the form $b_iA$ for some $b_i\in B$. In particular, $b_i\in g_iA$. Since $B\cong G/A$ is Abelian, $b_i$ commutes with $B$ and of course also with $A$, so $b_i\in Z(G)$. Since $\langle b_i\rangle\cap A$ is trivial, it follows that $|b_i|=|g_iA|$.

Conversely, suppose that for each $i \in [k]$, there exist $b_i \in Z(G) \cap g_{i}A$ such that $|b_{i}| = |g_{i}A|$. As each $b_i \in Z(G)$ ($i \in [k]$), we have that $A$ commutes with $\langle b_i \rangle$. Necessarily, we have that for all $i \in [k]$, $b_i \not \in \langle b_{1}, \ldots, b_{i-1}, b_{i+1}, \ldots, b_{k} \rangle$; otherwise, we would have $b_{i}A \in \langle b_{1}A, \ldots, b_{i-1}A, b_{i+1}A, \ldots, b_{k}A \rangle$, contradicting the fact that $G/A = \langle b_1A \rangle \times \cdots \times \langle b_k A \rangle$. Thus, $B = \prod_{i=1}^{k} \langle b_i \rangle$.
\end{proof}

\begin{proposition}[cf. {\cite[Section~4.2]{KayalNezhmetdinov}}] \label{prop:GroupDivision}
We can solve \algprobm{Group Division} in $\textsf{AC}^{1}$.
\end{proposition}

\begin{proof}
We consider two cases.
\begin{itemize}
\item \textbf{Case 1:} Suppose that $G/A$ is Abelian. Given $g, h \in G$, we can decide in $\textsf{AC}^{0}$ whether $g$ and $h$ belong to the same coset of $G/A$, by checking whether $gh^{-1} \in A$. Thus, we may write down the multiplication table for $G/A$ in $\textsf{AC}^{0}$. As $G/A$ is Abelian, we have by Proposition~\ref{prop:DecomposeAbelianGroup} that we have that we can compute a fully-refined direct product decomposition for $G/A = \prod_{i=1}^{k} \langle g_iA \rangle$ in $\textsf{AC}^{1}$. 

Now, for each coset $g_iA$ ($i \in [k]$), it remains to select a representative $b_i \in g_iA$ such that $b_i \in Z(G)$ and $|b_i| = |g_i A|$ if such $b_i$ exist. If for some $i$, no such $b_i$ exists, then the algorithm returns that $G$ does not admit a normal subgroup $B$ satisfying $G = A \times B$. The correctness of this step follows from Lemma~\ref{lem:AbelianDivision}.. By Lemma~\ref{lem:Order}, this step is computable in $\textsf{FL} \cap \FOLL$. The total work in this case is $\textsf{AC}^{1}$-computable, as desired.

\item \textbf{Case 2:} Suppose instead that $G/A$ is not Abelian. In $\ACz$, we may decide whether two elements of $G$ commute. Thus, in $\ACz$, we may write down the elements of $C_{G}(A)$. We now turn to computing $T = [C_{G}(A), G]$. In $\ACz$, we may write down $S := \{ [a,g] : a \in C_{G}(A), g \in G \}$. Now using a membership test, we may in $\textsf{FL}$, write down $T = \langle S \rangle$ \cite{BarringtonMcKenzie, Reingold}.

We now determine whether $T \trianglelefteq G$ in $\ACz$ by checking whether $gTg^{-1} = T$ for all $g \in G$. If $T$ is not normal in $G$, then we output that no such decomposition exists. We now check in $\ACz$ whether $A \cap T = \{1\}$. If $A \cap T \neq \{1\}$, then we output that no such decomposition exists. Let $\widetilde{A} := \{ Ta : a \in A \} \trianglelefteq G/T$. It remains to determine whether there exists some Abelian group $B/T \trianglelefteq G/T$ such that $G/T = \widetilde{A} \times B/T$; and if so, to construct such a $B/T = \langle b_1 T \rangle \times \cdots \times \langle b_k T \rangle$. By Case 1, we may accomplish this in $\textsf{AC}^{1}$. 

For each $i \in [k]$, we select a coset representative $c_i \in b_i T$; from Kayal and Nezhmetdinov \cite{KayalNezhmetdinov}, we have that any such coset representative suffices. Now in $\textsf{FL}$, we may write down $C := \langle T, c_1, \ldots, c_k \rangle$ using a membership test \cite{BarringtonMcKenzie, Reingold}. We may now check in $\ACz$ whether $G = A \times C$. If $G = A \times C$, then we return $C$. Otherwise, we return that no such decomposition exists. The total work in this case is $\textsf{AC}^{1}$-computable, as desired. 
\end{itemize}

We established the correctness of Case 1 above. The correctness of Case 2 was established in   \cite[Section~4]{KayalNezhmetdinov}. The result now follows.
\end{proof}

We now turn to the  \algprobm{SemiAbelianGroupDecomposition} problem, which takes as input a group $G$ and asks for a decomposition of the form:
\[
G = A \times \langle b_1 \rangle \times \cdots \times \langle b_k \rangle,
\]
where $A$ has no Abelian direct factors and $b_1, \ldots, b_k \in G$. Note that the direct product $B := \prod_{i=1}^{k} \langle b_i \rangle$ is Abelian.

In this section, we will establish the following.

\begin{proposition}[cf. {\cite[Section~5]{KayalNezhmetdinov}}] \label{prop:SemiAbelian}
We can solve \algprobm{SemiAbelianGroupDecomposition} in $\textsf{AC}^{2}$.   
\end{proposition}

\begin{proof}
If $G$ is Abelian, then $A = \{1\}$. By Proposition~\ref{prop:DecomposeAbelianGroup}, we can compute a fully-refined direct product decomposition for $G$ in $\textsf{AC}^{1}$. So suppose $G$ is non-Abelian.

For each $b \in G$, we apply \algprobm{GroupDivision} to determine if $\langle b \rangle$ admits a direct complement in $G$. We consider all such $b$ in parallel. By Proposition~\ref{prop:GroupDivision}, we can solve \algprobm{GroupDivision} in $\textsf{AC}^1$. If no such $b$ exists, we return $\{G\}$. Otherwise, we select one such $b$, and label it $b_1$. Let $C$ be the direct complement returned from \algprobm{GroupDivision} applied to $b_1$. So $G = C \times \langle b_1 \rangle$. We now recurse on C.  We perform at most $\lceil \log |G| \rceil$ recursive calls, each of which are $\textsf{AC}^{1}$-computable. Thus, the total work is $\textsf{AC}^2$-computable, as desired.
\end{proof}

\subsection{Conjugacy Class Graph of a Group}

In this section, we will leverage the conjugacy class graph of a group (to be defined shortly) to compute a fully-refined direct product decomposition of a non-Abelian group. We first recall key preliminaries  from \cite[Section~6]{KayalNezhmetdinov}. Let $G$ be a group. For $g \in G$, let $\mathcal{C}_{g}$ be the conjugacy class of $g$. We say that two conjugacy classes $\mathcal{C}_{g}, \mathcal{C}_{h}$ of $G$ \emph{commute} if for all $\alpha \in \mathcal{C}_{g}$ and all $\beta \in \mathcal{C}_{h}$, $\alpha$ and $\beta$ commute.

Suppose that $G = G_1 \times \cdots \times G_k$, and let $g = g_1 \cdots g_k$, with each $g_i \in G_i$ ($i \in [k]$). Then we have that:
\[
\mathcal{C}_{g} = \mathcal{C}_{g_{1}} \cdots \mathcal{C}_{g_{k}}.
\]
Furthermore, for each $i \neq j$, we have that $\mathcal{C}_{g_{i}}$ and $\mathcal{C}_{g_{j}}$ commute.

For the remainder of this section, we will assume that $G$ is non-Abelian (as we can use Proposition~\ref{prop:DecomposeAbelianGroup} to decompose an Abelian group). 

\begin{definition}[{\cite[Definition~2]{KayalNezhmetdinov}}]
Let $G$ be a group, and let $g \in G$. We say that the conjugacy class $\mathcal{C}_{g}$ is \emph{reducible} if either $g \in Z(G)$, or there exist two conjugacy classes $\mathcal{C}_{a}, \mathcal{C}_{b}$ such that the following hold:
\begin{itemize}
    \item $a, b \not \in Z(G)$,
    \item $\mathcal{C}_{a}, \mathcal{C}_{b}$ commute,
    \item $\mathcal{C}_{g} = \mathcal{C}_{a} \cdot \mathcal{C}_{b}$, and
    \item $|\mathcal{C}_{g}| = |\mathcal{C}_{a}| \cdot |\mathcal{C}_{b}|$.
\end{itemize}
\noindent If a conjugacy class is not reducible, then it is \emph{irreducible}.
\end{definition}

\begin{proposition}[{\cite[Proposition~2]{KayalNezhmetdinov}}] \label{prop:KNProp2}
Let $G = G_1 \times \cdots \times G_k$. If a conjugacy class $\mathcal{C}_{g}$ is irreducible, then there exists a unique $i \in [k]$ such that $\pi_{i}(g) \not \in Z(G_i)$, where $\pi_{i} : G \to G_{i}$ is the canonical projection map.
\end{proposition}

The converse of Proposition~\ref{prop:KNProp2} does not in general hold.

\begin{definition}
Let $G$ be a group. The \emph{conjugacy class graph} of $G$, denoted $\Gamma_{G}$, has precisely one vertex per irreducible conjugacy of $G$. For vertices $u, v$, we have that $\{u, v\} \in E(\Gamma_{G})$ if and only if the conjugacy classes corresponding to $u$ and $v$ do not commute.
\end{definition}

We will first show that we can compute $\Gamma_{G}$ and its connected components in logspace. 

\begin{lemma} \label{lem:ComputeConjugacyGraph}
Let $G$ be a group. We can compute $\Gamma_{G}$ in \textsf{FL}. Furthermore, we can compute the connected components of $\Gamma_{G}$ in $\textsf{FL}$.    
\end{lemma}

\begin{proof}
First, we observe that we can compute the conjugacy classes of $G$ in $\ACz$. We now claim that for a given conjugacy class $\mathcal{C}_{x}$, we can decide if $\mathcal{C}_{x}$ is irreducible in $\textsf{L}$. In $\ACz$, we can test whether $x \in Z(G)$. Furthermore, in $\ACz$, we may examine all pairs $a, b \in G$ and test whether $\mathcal{C}_{x} = \mathcal{C}_{a} \cdot \mathcal{C}_{b}$. Now in $\textsf{L}$, we may test whether $|\mathcal{C}_{x}| = |\mathcal{C}_{a}| \cdot |\mathcal{C}_{b}|$. Thus, in $\textsf{FL}$, we may identify the irreducible conjugacy classes of $G$.

We now turn to writing down $\Gamma_{G}$. The irreducible conjugacy classes of $G$ are precisely the vertices of $\Gamma_{G}$. Now let $\mathcal{C}_{a}, \mathcal{C}_{b}$ be two such irreducible conjugacy classes. The vertices corresponding to $\mathcal{C}_{a}, \mathcal{C}_{b}$ are not adjacent if and only if for every $g \in \mathcal{C}_{a}$ and every $h \in \mathcal{C}_{b}$, $gh = hg$. Thus, we can decide in $\ACz$ whether $\mathcal{C}_{a}, \mathcal{C}_{b}$ are adjacent in $\Gamma_{G}$. Once we have $\Gamma_{G}$, we may compute the connected components in $\textsf{FL}$ \cite{Reingold}. The result now follows.
\end{proof}

\begin{proposition}[{\cite[Proposition~5]{KayalNezhmetdinov}}] \label{prop:NumComponents}
Let $G$ be a group, and let $\Gamma_{G}$ be the conjugacy graph of $G$. The number of connected components of $\Gamma_{G}$ is at most $\log |G|$.    
\end{proposition}

\subsection{Algorithm}

We now turn to proving Theorem~\ref{thm:MainDecompose}.

\begin{proof}[Proof of Theorem~\ref{thm:MainDecompose}]
We now parallelize the main direct product decomposition algorithm of Kayal and Nezhmetdinov \cite[Algorithm~1]{KayalNezhmetdinov}. Let $G$ be our input group. If $G$ is Abelian, then we can decompose $G$ in $\textsf{AC}^{1}$ using Proposition~\ref{prop:DecomposeAbelianGroup}. So suppose $G$ is non-Abelian. We begin by computing the conjugacy class graph $\Gamma_{G}$, and the connected components $\Lambda_{1}, \ldots, \Lambda_{t}$ of $\Gamma_{G}$. By Lemma~\ref{lem:ComputeConjugacyGraph}, this step is $\textsf{FL}$-computable. Now if $G$ admits a non-trivial direct product decomposition, then there exists some $S_{1} \subseteq [t]$ inducing an indecomposable direct factor of $G$ \cite{KayalNezhmetdinov}. By Proposition~\ref{prop:NumComponents}, we have that $t \leq \log n$. Thus, we may try all such $S_{1}$ in parallel. 

For clarity, we will fix some $S_{1}$ and proceed as follows. For a connected component $\Lambda$ of $\Gamma_{G}$, let $\text{Elts}(\Lambda) \subseteq G$ be the set of group elements $g \in G$ such that $g$ belongs to some conjugacy class occurring in $\Lambda.$ Let
\[
Z_{1} := C_{G} \left( \left \langle \bigcup_{j \not \in S_{1}} \text{Elts}(\Lambda_{j}) \right \rangle \right).
\]

We may write down the elements of $\text{Elts}(\Lambda_{j})$ ($j \not \in S_{1}$) in $\ACz$. We may then use a membership test to, in $\textsf{FL}$, write down:
\[
X := \left \langle \bigcup_{j \not \in S_{1}} \text{Elts}(\Lambda_{j}) \right \rangle
\]
Given $X$, we may write down $Z_{1} = C_{G}(X)$ in $\ACz$. Thus, we may write down $Z_{1}$ in $\textsf{FL}$. Using Proposition~\ref{prop:SemiAbelian}, we can in $\textsf{AC}^{2}$, determine $H_{1}, Y \trianglelefteq G$ such that $Z_{1} = H_{1} \times Y$, where $H_{1}$ has no Abelian direct factors and $Y$ is Abelian. We now use Proposition~\ref{prop:GroupDivision} to decide if there exists $Y_{1} \trianglelefteq G$ such that $G = H_{1} \times Y_{1}$. If no such $Y_{1}$ exists, then $H_{1}$ is not an indecomposable direct factor of $G$. Otherwise, we recursively compute a direct product decomposition of $Y_{1}$.

If for all choices of $S_{1}$, no such $Y_{1}$ exists, then $G$ is directly indecomposable.

The correctness of the algorithm was established in \cite{KayalNezhmetdinov}. It remains to analyze the complexity. Note that we make at most $O(\log |G|)$ recursive calls (to compute a direct product decomposition of the subgroup $Y$ at the current recursive call). Each recursive call is $\textsf{AC}^{2}$-computable. Thus, our algorithm is $\textsf{AC}^{3}$-computable, as desired.
\end{proof}

\section{Canonizing Direct Products in Parallel}

In this section, we will prove Theorem~\ref{thm:MainCanonization}. We first recall the following lemma, which is essentially in  Levet's dissertation \cite{LevetThesis}. For completeness, we also include a proof.

\begin{lemma}[cf. {\cite{LevetThesis}}] \label{lem:CanonizeBoundedGenerator}
Let $d > 0$ be a constant, and let $G$ be a group that is generated by at most $d$ elements. Let $\chi_{d,0}$ be the coloring computed by $(d,0)$-WL Version II, applied to $G$. We can compute a canonical labeling for $G$ and return a generating sequence $(g_1, \ldots, g_d)$ of minimum color class from $G$ under $\chi_{d,0}$ that induces said labeling, in $\textsf{FL}$. 
\end{lemma}

\begin{proof}
For each $d$-tuple $(g_1, \ldots, g_d) \in G^d$, we check in $\textsf{L}$ whether $(g_1, \ldots, g_d)$ generates $G$ \cite{TangThesis}. We take such a $d$-tuple that belongs to the minimum such color class from $\chi_{d,0}$ (under the natural ordering of the labels of the color classes). 

In order to obtain a canonical labeling, we individualize each non-identity element from $g_1, \ldots, g_d$ such that for all distinct $i, j \in [d]$, $g_{i}$ and $g_{j}$ receive different colors. We then run $(d, O(1))$-WL, which we again note is $\textsf{FL}$-computable. Each color class will now have size $1$. We label an element $g \in G$ according to the color class of the $d$-tuple $(g, \ldots, g)$. 

Now suppose that $H \cong G$, and let $(h_1, \ldots, h_d) \in H^d$ be the $d$-tuple obtained by the above procedure. As $(g_1, \ldots, g_d)$ was obtained from the minimum color class and the WL-coloring is an isomorphism invariant, we have necessarily that the map $(g_1, \ldots, g_d) \mapsto (h_1, \ldots, h_d)$ extends to an isomorphism of $G \cong H$. The result now follows.
\end{proof}

\begin{proof}[Proof of Theorem~\ref{thm:MainCanonization}]
We first use Theorem~\ref{thm:MainDecompose} to, in $\textsf{AC}^{3}$, decompose $G = G_{1} \times \cdots \times G_{k}$ such that for each $i \in [k]$, $G_{i}$ is directly indecomposable. Now for each $i \in [k]$, we may verify in $\textsf{L}$ whether $G_{i}$ is generated by at most $d$ elements \cite{TangThesis}. If some $G_{i}$ requires more than $d$ generators, we reject.

Now we run $(d,0)$-WL Version II on $G$. For each $i \in [k]$, we may in $\ACz$, identify the colored $d$-tuples of $G_{i}^{d}$. As only the initial coloring of WL has occurred, observe that the partition on $G_{i}^{d}$ obtained by these colored $d$-tuples agrees exactly with that of applying $(d,0)$-WL Version II to $G_{i}$. Thus, we may apply Lemma~\ref{lem:CanonizeBoundedGenerator} to obtain a canonical labeling $\lambda_{i}$ of each $G_{i}$ and a canonical generating sequence $(g_{i,1}, \ldots, g_{i,d})$ in $\textsf{FL}$. In particular, Lemma~\ref{lem:CanonizeBoundedGenerator} guarantees that $\chi_{d,0}((g_{i,1}, \ldots, g_{i,d}))$ is the minimum color class amongst all such generating sequences in $G_{i}^{d}$.

We now sort $G_{1}, \ldots, G_{k}$ such that $G_{i}$ appears before $G_{j}$ whenever $\chi_{d,0}((g_{i,1}, \ldots, g_{i,d})) < \chi_{d,0}((g_{j,1}, \ldots, g_{j,d}))$. By re-indexing $G_{1}, \ldots, G_{k}$, we may assume (without loss of generality) that the direct factors appear in order $(G_{1}, \ldots, G_{k})$. This step is $\textsf{FL}$-computable. In order to obtain a canonical labeling, we now individualize the non-identity elements of our canonical tuples (without loss of generality, we may assume that all such elements are not the identity):
\[
(g_{1,1}, \ldots, g_{1,d}, g_{2,1}, \ldots, g_{2,d}, \ldots, g_{k,1}, \ldots, g_{k,d}),
\]
such that for all distinct pairs $(i, j), (i', j') \in [k] \times [d]$, $g_{i,j}$ and $g_{i', j'}$ are individualized to receive a different color. We then run $(O(1), O(1))$-WL Version II on the original uncolored group $G$, which is $\textsf{FL}$-computable. Observe that for each $i \in [k]$, the resulting coloring induces the labeling $\lambda_{i}$ on $G_{i}$, and hence subsequently assigns each element in $G$ a unique color. Precisely, we assign to an element $g \in G$ the label corresponding to the color class $\chi_{O(1), O(1)}((g, \ldots, g))$.

Thus, in total, our algorithm is $\textsf{AC}^{3}$-computable. It remains to argue that our labeling is canonical. Let $H \cong G$. Thus, we may write $H = \prod_{i=1}^{k} H_{i}$, where for each $i \in [k]$, $H_{i}$ is directly indecomposable. As $H \cong G$, there exists a permutation $\pi \in \text{Sym}(k)$ such that $G_{i} \cong H_{\pi(i)}$ for all $i \in [k]$. As the Weisfeiler--Leman colorings are invariant under isomorphism and as $H \cong G$, our algorithm will sort the direct factors of $H$ such that $G_{i} \cong H_{\pi(i)}$. 

Now by Lemma~\ref{lem:CanonizeBoundedGenerator}, the canonical generating sequence $(g_{i,1}, \ldots, g_{i,d})$ for $G_{i}$ is taken to belong to the minimum color class under the WL coloring. As the WL coloring is isomorphism invariant, the canonical generating sequence $(h_{\pi(i), 1}, \ldots, h_{\pi(i), d})$ for $H_{\pi(i)}$ will belong to the same color class as $(g_{i,1}, \ldots, g_{i,d})$. By the definition of the initial coloring for WL Version II, $(g_{i,1}, \ldots, g_{i,d})$ and $(h_{\pi(i), 1}, \ldots, h_{\pi(i), d})$ receive the same initial color if and only if the map $(g_{i,1}, \ldots, g_{i,d}) \mapsto (h_{\pi(i), 1}, \ldots, h_{\pi(i), d})$ extends to an isomorphism of the generated subgroups. Hence, our algorithm will induce the same labeling on $G_{i}$ and $H_{\pi(i)}$. 

Now we have that our algorithm sorts the direct factors so that, without loss of generality, for all $i \in [k]$, $G_{i}$ and $H_{\pi(i)}$ appear in the same position. Furthermore, we have that our algorithm induces the same labeling on $G_{i}$ and $H_{\pi(i)}$ for all $i \in [k]$. Thus, our algorithm induces the same labeling on $G$ and $H$. The result now follows.
\end{proof}

We also show that when $G \in \mathcal{C}$ (each indecomposable direct factor is $O(1)$-generated, and either perfect or centerless), we can compute a canonical labeling in logspace. 

\begin{theorem}
Let $G \in \mathcal{C}$. We can compute a canonical labeling for $G$ in $\textsf{FL}$.    
\end{theorem}

\begin{proof}
Let $d \in \mathbb{N}$ such that each indecomposable direct factor of $G$ is generated by at most $d$ elements (note that such a $d$ exists as $G \in \mathcal{C}$). We first show how to compute a fully-refined direct product decomposition of $G$ in $\textsf{FL}$. For each $d$-tuple $(g_1, \ldots, g_d) \in G^d$ of distinct elements, we use a membership test to write down the elements of $H := \langle g_1, \ldots, g_d \rangle$ in $\textsf{FL}$ \cite{BarringtonMcKenzie, Reingold}. We may now check in $\ACz$ whether $H$ is normal in $G$. If $H$ is not normal in $G$, then $H$ is not a direct factor of $G$. We may also check in $\ACz$ whether $Z(H) = \{1\}$. 

If $Z(H) \neq \{1\}$, then we check in $\textsf{FL}$ whether $H$ is perfect, in the following manner. We first write down, in $\ACz$,  $[x,y]$, for all $x, y \in H$. Now using a membership test, we can write down $[H, H]$ in $\textsf{FL}$ \cite{BarringtonMcKenzie, Reingold} and then check if $[H,H] = H$. 

It follow that, in $\textsf{FL}$, we can (i) decide if $G \in \mathcal{C}$, and if so (ii) compute a fully refined direct product decomposition of $G$. The remainder of the proof is now identical to that of Theorem~\ref{thm:MainCanonization}.    
\end{proof}

\section{Central Quasigroups in NC}

In this section, we will establish the following.

\begin{theorem} \label{thm:CentralQuasigroups}
Let $(G_1, *_{1})$ be a central quasigroup, and let $(G_2, *_{2})$ be an arbitrary quasigroup. Suppose that $G_1, G_2$ are given by their multiplication tables. We can decide isomorphism between $G_1$ and $G_2$ in $\textsf{NC}$.    
\end{theorem}

We will begin by recalling some additional preliminaries.

\subsection{Additional Preliminaries: Central Quasigroups}

\begin{lemma} \label{lem:CommutingAutomorphisms}
Let $(Q,*)$ be a central quasigroup over an Abelian group $(Q,+)$. Then there are uniquely determined $\phi,\psi\in\mathrm{Aut}(Q,+)$ and $c\in Q$ such that $(Q,*)=\mathcal Q(Q,+,\phi,\psi,c)$. Namely, if $0$ is the identity element of $(Q,+)$, then $c=0*0$, $\phi(x)=(x*0)-(0*0)$ and $\psi(x) = (0*x)-(0*0)$. Furthermore, if $Q$ is given by its multiplication table, we can recover $\phi, \psi$ in $\textsf{AC}^{0}$.
\end{lemma}

\begin{proof}
Let $(Q,*)=\mathcal Q(Q,+,\phi,\psi,c)$. Then $0*0 = \phi(0)+\psi(0)+c = c$, $x*0 = \phi(x)+\psi(0)+c = \phi(x)+(0*0)$, and $0*y = \phi(0)+\psi(y)+c = \psi(y)+(0*0)$. 

As $Q$ is given by its multiplication table, we can evaluate the expressions for $\phi(x), \psi(x)$ in $\textsf{AC}^{0}$. 
\end{proof}

In \cite{StanovskyVojtechovsky}, the authors provided a test of central quasigroups up to isomorphism. We will crucially leverage this result to construct an isomorphism test. 

\begin{theorem}[{\cite[Theorem~2.4]{StanovskyVojtechovsky}}] \label{thm:SV}
Let $(G,+)$ be an Abelian group, let $\phi_{1}, \psi_{1}, \phi_{2}, \psi_{2} \in \Aut(G,+)$, and let $c_{1}, c_{2} \in G$. Then the following statements are equivalent:
\begin{enumerate}[label=(\alph*)]
\item the central quasigroups $\mathcal{Q}(G, +, \phi_{1}, \psi_{1}, c_{1})$ and $\mathcal{Q}(G, +, \phi_{2}, \psi_{2}, c_{2})$ are isomorphic, 
\item there is an automorphism $\gamma \in \Aut(G,+)$ and an element $u \in \text{Im}(1-\phi_{1}-\psi_{1})$ such that
\[
\phi_{2} = \gamma \phi_{1} \gamma^{-1}, \, \psi_{2} = \gamma \psi_{1} \gamma^{-1}, \, c_{2} = \gamma(c_{1} + u).
\]
\end{enumerate}
\end{theorem}

Recall that a group $G$ acts on a set $X$ \emph{regularly} if for every $x,y\in X$ there is a unique $g\in G$ such that $g\cdot x = y$.

\begin{proposition} \label{prop:IsotopicGroup}
Let $(Q,*)$ be a quasigroup. Then $(Q,*)$ is isotopic to a group if and only if $\dis{Q,*}$ acts regularly on $Q$, in which case $(Q,*)$ is isotopic to $\dis{Q,*}$.
\end{proposition}

\begin{proof}
Suppose that $(Q,\cdot)$ is a group and $(\alpha,\beta,\gamma)$ is an isotopism from $(Q,*)$ to $(Q,\cdot)$. Then $x*y = \gamma^{-1}(\alpha(x)\cdot\beta(y))$. Denote by $L_x$ the left translation by $x$ in $(Q,\cdot)$, and by $\lambda_x$ the left translation by $x$ in $(Q,*)$. We have $\lambda_x = \gamma^{-1}L_{\alpha(x)}\beta$. Let $e\in Q$. Then $\lambda_x\lambda_e^{-1} = (\gamma^{-1}L_{\alpha(x)}\beta)(\gamma^{-1}L_{\alpha(e)}\beta)^{-1} = \gamma^{-1}L_{\alpha(x)}L_{\alpha(e)}^{-1}\gamma = \gamma^{-1}L_{\alpha(x)\alpha(e)^{-1}}\gamma$ since $(Q,\cdot)$ is a group. Hence $\dis{Q,*} = \langle\lambda_x\lambda_e^{-1}:x\in Q\rangle = \langle \gamma^{-1}L_{\alpha(x)\alpha(e)^{-1}}\gamma:x\in Q\rangle$. Consequently, $\gamma\dis{Q,*}\gamma^{-1} = \langle L_{\alpha(x)\alpha(e)^{-1}}:x\in Q\rangle = \langle L_x:x\in Q\rangle = \{L_x:x\in Q\}$, where in the last equality we again used the fact that $(Q,\cdot)$ is a group. As the group $\{L_x:x\in Q\}$ acts regularly on $Q$, it follows that the conjugate subgroup $\dis{Q,*}$ also acts regularly.

Conversely, suppose that $\dis{Q,*}$ acts regularly on $Q$. We claim that $\dis{Q,*}=\{\lambda_x\lambda_e^{-1}:x\in Q\}$. By regularity, it suffices to show that for every $y,z\in Q$ there is $x\in Q$ such that $\lambda_x\lambda_e^{-1}(y)=z$. Now, $\lambda_x\lambda_e^{-1}(y)=z$ if and only if $x*(e\ldiv y) = z$ if and only if $x = z\rdiv (e\ldiv y)$.

Since $\dis{Q,*}=\{\lambda_x\lambda_e^{-1}:x\in Q\}$ is a group under the composition of mappings, for every $x,y\in Q$ there is a unique $z\in Q$ such that $\lambda_x\lambda_e^{-1}\lambda_y\lambda_e^{-1}=\lambda_z\lambda_e^{-1}$. This is equivalent to $\lambda_z = \lambda_x\lambda_e^{-1}\lambda_y$, and applying these mappings to $e$ yields $z*e = x*(e\ldiv (y*e))$, that is, $z = (x*(e\ldiv (y*e)))\rdiv e$. Letting $x\cdot y = (x*(e\ldiv (y*e)))\rdiv e$, it follows that $\dis{Q,*}$ is isomorphic to $(Q,\cdot)$ and, in particular, $(Q,\cdot)$ is a group. Let $\alpha,\beta,\gamma:Q\to Q$ be the permutations defined by $\alpha(x)=x$, $\beta(y)=e\ldiv(y*e)$ and $\gamma(z)=z*e$. Then $\gamma(x\cdot y) = (x\cdot y)*e = x*(e\ldiv (y*e)) = \alpha(x)*\beta(y)$, and hence $(Q,*)$ is isotopic to $(Q,\cdot)$.
\end{proof}

\begin{lemma} \label{lem:CentralQuasigroupIsotopy}
Let $(G, +)$ be an Abelian group, and let $\phi, \psi \in \Aut(G, +)$. Let $c \in G$. Let $\mathcal{Q}(G, +, \phi, \psi, c)$ be the corresponding central quasigroup, with multiplication $x * y = \phi(x) + \psi(y) + c$. We have that $\mathcal{Q}(G, +, \phi, \psi, c)$ is isotopic to $(G, +)$.
\end{lemma}

\begin{proof}
We claim that $(u,v,w) = (\phi, \psi, x\mapsto x-c)$ is an isotopism from $(G, +)$ to $\mathcal{Q}(G, +, \phi, \psi, c)$. Indeed, $u(x)+v(y) = w(x*y)$ is equivalent to $\phi(x)+\psi(y) = (x*y)-c$, which holds by definition of $*$. 
\end{proof}

\subsection{Pointwise Transporter on Small Domains in Parallel}

In the permutation group model, a group $G$ is specified succinctly by a generating sequence of permutations from $\text{Sym}(m)$. The computational complexity for this model will be measured in terms of $m$. We will recall the following standard suite of $\textsf{NC}$ algorithms in the setting of permutation groups.

\begin{lemma} \label{lem:PermutationGroupsNC}
Let $G \leq \text{Sym}(m)$ be given by a sequence $S$ of generators. The following problems are in $\textsf{NC}$:
\begin{enumerate}[label=(\alph*)]
\item Find the pointwise stabilizer of $B \subseteq [m]$ \cite{BabaiLuksSeress}.
\item For any $x \in [m]$, find the orbit $x^{G}$. Furthermore, for each $y \in x^{G}$, find a $g \in G$ such that $x^{g} = y$ \cite{McKenzieThesis}.
\end{enumerate}
\end{lemma}

The \algprobm{Pointwise Transporter} problem takes as input $G \leq \text{Sym}(m)$ given by a sequence of generators, and two sequences $(x_1, \ldots, x_k)$, $(y_1, \ldots, y_k)$ of points from $[m]$, and tests whether there exists some $g \in G$ such that for all $i \in [k]$, $x_{i}^{g} = y_{i}$. We will establish the following:

\begin{lemma} \label{lem:PointwiseTransporter}
Fix $c > 0$, and let $k \in O(\log^{c} n)$. We can solve \algprobm{Pointwise Transporter} on $k$ points in $\textsf{NC}$.
\end{lemma}

\begin{proof}
We follow the strategy of \cite[(3.9)]{LuksReduction}. By Lemma~\ref{lem:PermutationGroupsNC}(b), we can decide in $\textsf{NC}$ if there exists some $g \in G$ such that $x_{1}^{g} = y_{1}$. If so, we compute $\text{Stab}(x_1)$ in $\textsf{NC}$ (Lemma~\ref{lem:PermutationGroupsNC}(a)). We now recursively find the subcoset $Hz$ of $\text{Stab}(x_1)$ sending $x_{i} \mapsto y_{i}^{g^{-1}}$ for all $2 \leq i \leq k$. 

We make $k \in O(\log^{c} n)$ recursive calls, each of which is $\textsf{NC}$-computable. Thus, our entire algorithm is $\textsf{NC}$-computable, as desired.
\end{proof}

\subsection{Isomorphism Testing of Central Quasigroups}

\begin{lemma} \label{lem:DisQ}
Let $(Q, *)$ be a quasigroup, and fix $e \in Q$ arbitrarily. We can write down the set $D = \{ \lambda_{x}\lambda_{e}^{-1} : x \in Q\} \subseteq \dis{Q,*}$ and check if $(Q,*)$ is isotopic to a group; and if so, then we can also write down the multiplication table for $\dis{Q,*}$ all in $\ACz$.
\end{lemma}

\begin{proof}
We have that $\lambda_{x}\lambda_{e}^{-1}(y) = x * (e\backslash y)$, which we may evaluate in $\textsf{AC}^{0}$. Thus, in $\textsf{AC}^{0}$, we may evaluate in parallel $\lambda_{x}\lambda_{e}^{-1}(y)$, for all $x, y \in Q$. It follows that we may write down $D$ in $\ACz$. 

By Proposition~\ref{prop:IsotopicGroup}, $(Q,*)$ is isotopic to a group if and only if $\dis{Q,*}$ acts regularly on $Q$. Now, $\dis{Q,*}$ acts regularly on $Q$ if and only if $|\dis{Q,*}|=|Q|$. Since $|D|=|Q|$, $\dis{Q,*}$ acts regularly on $Q$ if and only if $D=\dis{Q,*}$. This last condition holds precisely when $D$ is closed under composition, which can be checked in $\ACz$.

By Theorem~\ref{thm:SV}, it remains to write down the multiplication table for $\dis{Q,*}$. Now for $z = (x*(e\ldiv (y*e)))\rdiv e$, we have that $\lambda_{x}\lambda_{e}^{-1} \cdot \lambda_{y}\lambda_{e}^{-1} = \lambda_{z}\lambda_{e}^{-1}$. Define $x\cdot y = (x*(e\ldiv (y*e)))\rdiv e$. We may thus, in $\textsf{AC}^{0}$, write down the group $(Q,\cdot)$, which is isomorphic to $\dis{Q, *}$ (Proposition~\ref{prop:IsotopicGroup}). The result now follows.
\end{proof}

\begin{proposition} \label{prop:RecognizeCentral}
Let $(Q,*)$ be a quasigroup. We can decide in $\textsf{AC}^{0}$ whether $(Q,*)$ is central and construct its representation $\mathcal Q(Q,+,\phi,\psi,c)$.
\end{proposition}

\begin{proof}
We first use Lemma~\ref{lem:DisQ} to check, in $\ACz$, if $(Q,*)$ is isotopic to a group. If so, Lemma~\ref{lem:DisQ} allows us to, in $\ACz$, construct the multiplication table for the underlying group, which we denote $(Q, \cdot)$. We note that $(Q,\cdot)$ is isomorphic to $\mathrm{Dis}(Q,*)$. Now we may, in $\ACz$, verify that $(Q, \cdot)$ is Abelian. By Lemma~\ref{lem:CentralQuasigroupIsotopy}, every central quasigroup $(Q,*)$ is isotopic to $\dis{Q,*}$. Thus, if $(Q, \cdot)$ is not Abelian, then $(Q, *)$ is not central, and so we reject.

Thus, at this stage, suppose that $(Q, \cdot)$ is Abelian, and we refer to the group operation as addition ($+$) rather than multiplication ($\cdot$). By Lemma~\ref{lem:CommutingAutomorphisms}, we test in $\textsf{AC}^{0}$ whether the corresponding $c$ and automorphisms $\phi, \psi \in \Aut(Q, +)$, as prescribed by Theorem~\ref{thm:SV}, exist. We may then test, in $\ACz$, whether for all $x, y \in Q$,
$x * y = \phi(x) + \psi(y) + c$, as prescribed by Definition~\ref{def:CentralQuasigroup}. The result now follows. 
\end{proof}

We now turn to proving Theorem~\ref{thm:CentralQuasigroups}.

\begin{proof}[Proof of Theorem~\ref{thm:CentralQuasigroups}]
By Proposition~\ref{prop:RecognizeCentral}, we can decide in $\textsf{AC}^{0}$ whether both $G_1, G_2$ are central; and if so, construct the Abelian groups $(G_i,+)\cong\dis{G_{i}, *_i}$ and their representations $\mathcal{Q}_{i}(G_i, +, \phi_i, \psi_i, c_i)\cong G_i$. Now using Proposition~\ref{prop:DecomposeAbelianGroup}, we may in $\textsf{AC}^{1}$ compute a basis $b_{i1}, \ldots, b_{ik}$ for $\dis{G_i, *_i}$, and sort the elements so that $|b_{ij}| \leq |b_{i,j+1}|$ for all $i, j$. Note that $k \leq \log |G_i|$. In particular, this allows us to decide if $\dis{G_1, *_i} \cong \dis{G_2, *_2}$ are isomorphic; and if so, construct an isomorphism $\tau : \dis{G_1, *_1} \cong \dis{G_2, *_2}$. Let $(G, +)$ be the Abelian group such that $(G, +) \cong \dis{G_i, *_i}$.

Now let:
\begin{align*}
&\bar{x} := (c_1, \phi_{1}(b_{11}), \ldots, \phi_{1}(b_{1k}), \psi_{1}(b_{11}), \ldots, \psi_{1}(b_{1k})), \\
&\bar{y} := (c_2, \phi_{2}(b_{21}), \ldots, \phi_{2}(b_{2k}), \psi_{2}(b_{21}), \ldots, \psi_{2}(b_{2k})).
\end{align*}

For $u \in (G, +)$ (precisely, we take $u \in \dis{G_1, *_1}$), let:
\[
\bar{x}_{u} = (c_1 + u, \phi_{1}(b_{11}), \ldots, \phi_{1}(b_{1k}), \psi_{1}(b_{11}), \ldots, \psi_{1}(b_{1k})).
\]

By Theorem~\ref{thm:SV}, it remains to decide whether there exists some $u \in \text{Im}(1-\phi_{1} - \psi_{1})$ and an automorphism $\gamma \in \Aut(G, +)$ such that $\bar{x}_{u}^{\gamma} = \bar{y}$. Thus, we have reduced to the \algprobm{Pointwise Transporter} problem. We will in fact, show how to obtain all such $\gamma \in \Aut(G, +)$. Let $\text{Trans}(\bar{x}_{u}, \bar{y})$ be the set of $\gamma \in \Aut(G,+)$ such that $\bar{x}_{u}^{\gamma} = \bar{y}$. 

First, as we are given $\phi_{1}, \psi_{1}$, we may in $\ACz$ compute $\text{Im}(1-\phi_{1} - \psi_{1})$ by evaluating $(1-\phi_{1}-\psi_{1})(u)$ for all $u \in (G,+)$. Now in $\textsf{NC}$, we may write down a set of $\poly(n)$ generators for $\Aut(G, +)$ \cite{Birkhoff}. Note that while Birkhoff \cite{Birkhoff} prescribes matrix representations for the generators, we may easily obtain permutations by considering the actions of the matrices on the basis $b_{11}, \ldots, b_{1k}$ for $(G, +)$.

As we are considering the \algprobm{Pointwise Transporter} problem on $2k+1 \leq 2\lceil \log(n) \rceil + 1$ points, we have by Lemma~\ref{lem:PointwiseTransporter} the requisite instances of \algprobm{Pointwise Transporter} are $\textsf{NC}$-computable. If there exists some $u \in \text{Im}(1-\phi_{1}-\psi_{1})$ such that $\text{Trans}(\bar{x}_{u}, \bar{y}) \neq \emptyset$, then our two quasigroups $G_{1}$ and $G_{2}$ are isomorphic. Otherwise, we have that $G_{1} \not \cong G_{2}$.

In order to recover the set of automorphisms $\gamma \in \Aut(G,+)$ as in Theorem~\ref{thm:SV}, we return:
\[
\bigcup_{u \in \text{Im}(1-\phi_{1}-\psi_{1})} \text{Trans}(\bar{x}_{u}, \bar{y}).
\]
The result now follows.
\end{proof}

\section{Conclusion}
In this paper, we investigated the Weisfeiler--Leman dimension of groups that admit (unique) direct product decompositions, where each indecomposable direct factor is $O(1)$-generated, and either perfect or centerless. We showed that the count-free $(O(1), O(\log \log n))$-WL Version II and counting $(O(1), O(1))$-WL Version II algorithms identify every such group. Consequently, we obtained upper bounds of $\textsf{L}$ for isomorphism testing between a group in this family, and an arbitrary group.

We also exhibited an $\textsf{AC}^{3}$ canonization procedure for the class $\mathcal{D}$ of groups where each indecomposable direct factor is $O(1)$-generated.

Finally, we obtained bounds of $\textsf{NC}$ for isomorphism testing of central quasigroups. Previously, only the trivial bound of $n^{\log(n) +O(1)}$-time was known.

Our work leaves several open questions.

\begin{question} \label{q:Q1}
Consider the subclass $\mathcal{D}'$ of $\mathcal{D}$ where each group in $\mathcal{D}'$ has a unique direct product decomposition. Does $\mathcal{D}'$ have bounded count-free Weisfeiler--Leman dimension?
\end{question}

Note that $\mathcal{D}$ includes the class of all Abelian groups. Grochow and Levet \cite{GrochowLevetWL} showed that the class of all Abelian groups has count-free WL dimension $\Theta(\log n)$.

\begin{question} \label{q:Q2}
Is $\mathcal{D}$ identified by the counting $(O(1), O(1))$-WL Version II?
\end{question}

Grochow and Levet \cite{GrochowLevetWL} showed that $(O(1), O(\log n))$-WL Version II identifies $\mathcal{D}$. It is unclear how to approach this, even in the restricted case when each group admits a unique fully-refined direct product decomposition. 

\begin{question} \label{q:Q3}
In the setting of groups specified by their multiplication tables, can we compute fully-refined direct product decompositions in $\textsf{FL}$?    
\end{question}

We are unable to resolve Question~\ref{q:Q3} even for Abelian groups-- see Proposition~\ref{prop:DecomposeAbelianGroup}, which yields bounds of $\textsf{AC}^{1}$.

\begin{question}
Can we compute a fully-refined direct product decomposition for quasigroups, specified by their multiplication tables, in polynomial-time?    
\end{question}

\begin{question}
Does isomorphism testing of central quasigroups belong to $\LogSpace$?    
\end{question}

Note that isomorphism testing of Abelian groups belongs to $\LogSpace$; this bound has been improved to \\ $\forall^{\log \log n}\textsf{MAC}^{0}(\DTISPpll)$, which is a proper subclass of $\textsf{L}$ \cite{CGLWISSAC}. As central quasigroups are precisely the Abelian objects (in the sense of universal algebra) in the variety of quasigroups \cite{Szendrei}, $\textsf{L}$ is a natural target.

\section*{Acknowledgements}

We thank the anonymous referees for helpful feedback. ML thanks Joshua A. Grochow for helpful discussions. DJ was partially supported by the Department of Computer Science at the College of Charleston, SURF Grant SU2024-06 from the College of Charleston, and ML startup funds. ML was partially supported by SURF Grant SU2024-06 from the College of Charleston. PV supported by the Simons Foundation Mathematics and Physical Sciences Collaboration Grant for Mathematicians no.~855097. BW was partially supported by a research assistantship from the Department of Computer Science at the College of Charleston, as well as ML startup funds.

\bibliographystyle{alphaurl}
\bibliography{references}

\end{document}